\documentclass[onecolumn,11pt]{IEEEtran}

\usepackage{amsfonts,amsmath,amssymb,amsthm,caption,extarrows}
\usepackage[linesnumbered,boxed,ruled,commentsnumbered]{algorithm2e}
\usepackage{amsbsy}
\usepackage{graphicx,multirow,bm}
\usepackage{color}
\usepackage{subfigure}
\usepackage{pifont}

\newtheoremstyle{mythm}{3pt}{3pt}{}{16pt}{\bfseries}{:}{.5em}{}
\theoremstyle{mythm}
\newtheorem{theorem}{Theorem}
\newtheorem{example}{Example}
\newtheorem{definition}{Definition}
\newtheorem{remark}{Remark}

\newtheorem{proposition}{Proposition}

\newtheorem{corollary}{Corollary}
\newtheorem{lemma}{Lemma}

\newcommand{\tabincell}[2]{\begin{tabular}{@{}#1@{}}#2\end{tabular}}
\begin{document}
\title{On Anti-collusion Codes for Averaging Attack in Multimedia Fingerprinting
\author{~Jing Jiang, ~Cailin Wen, ~Minquan Cheng}
\thanks{Jing Jiang, Cailin Wen and Minquan Cheng are with Key Lab of Education Blockchain and Intelligent Technology, Ministry of Education and Guangxi Key Lab of Multi-Source Information Mining and Security,
Guangxi Normal University, Guilin, 541004, China (e-mail: jjiang2008@hotmail.com, wen20190225@163.com, chengqinshi@hotmail.com).}
}

\maketitle
%





\vskip 1cm

\begin{abstract}
Multimedia fingerprinting is a technique to protect the copyrighted contents against being illegally redistributed under various collusion attack models.
Averaging attack is the most fair choice for each colluder to avoid detection,
and also makes the pirate copy have better perceptional quality.
This makes such an attack one of the most feasible approaches to carrying out collusion.
In order to trace all the colluders, several types of multimedia fingerprinting codes
were introduced to  construct fingerprints resistant to averaging attacks on multimedia contents,
such as  AND anti-collusion codes (AND-ACCs), binary separable codes (SCs), logical anti-collusion codes (LACCs),
binary frameproof  codes (FPCs),  binary strongly-separable codes (SSCs) and
binary secure code with list decoding (SCLDs).
Then codes with the rate as high as possible are desired.
However, the existing fingerprinting codes have low code rate due to the strong combinatorial structure.
The reason is that the previous research methods adopted simple tracing algorithms.
In this paper, we first  propose novel tracing algorithms
and then find appropriate fingerprinting codes with weaker combinatorial structure,
i.e.,  the binary strongly identifiable parent property code for multimedia fingerprinting (SMIPPC)
and  its concatenated code.
Theoretical comparisons and numerical comparisons show that
SMIPPCs have higher code rates than those of the existing codes
due to their weaker combinatorial structures.
It is worth noting that SMIPPCs can only trace a part of colluders by using the previous tracing algorithm
and the concatenated SMIPPC may be not an SMIPPC.
This implies that our tracing algorithms have strong traceability.
\end{abstract}




\section{Introductions}

Multimedia contents, such as video, audio, image, can be copied and distributed easily,
especially in the Internet age.
The illegal redistribution of copyrighted contents damages the interests of copyright owners. 
It is desired to devise techniques for copyright protection of multimedia contents.

Fingerprinting techniques which by providing unique identification of data in a certain manner
can be used to fight against illegal redistribution of copyrighted contents.
Clearly, an individual user cannot redistribute his/her copy  without running the risk of being tracked down.
The global nature of the Internet has also brought adversaries
closer to each other, and it is easy for a group of authorized users with differently
marked versions of the same content to mount attacks against the fingerprints.

Linear collusion is one of the most feasible collusion attacks against multimedia fingerprinting.
Since no colluder wishes to take more of a risk than any other colluder,
the marked versions are averaged with an equal weight for each colluder \cite{EKK,KLMSTZ,Sto,SEG,TWL}.
Such an attack is called averaging attack.
In this case, the trace of each individual fingerprint becomes weaker as the number of colluders increases.
In addition, the colluded signal can have better perceptual quality in that it can be more similar to the host signal
than the fingerprinted signals are.



In order to resist collusion attacks,
an appropriate fingerprinting code, which is a set of vectors (each vector represents an authorized user's fingerprint) with desired properties, and
a corresponding tracing algorithm are required.
$t$-resilient AND anti-collusion code ($t$-AND-ACC) was proposed by Trappe {\it et al.} \cite{TWL,TWWL}
to construct fingerprinted signals to resist averaging attack,
and the tracing algorithm based on $t$-AND-ACC was also proposed in \cite{TWL,TWWL} to detect up to $t$ colluders taking part in the attack.
Several constructions for $t$-AND-ACCs can be found in \cite{DSSSU,LWLPF,LT}.
Later in 2011, Cheng and Miao \cite{CM} introduced logical anti-collusion code (LACC) where not only the logical AND operation but also
the logical OR operation is exploited, and designed the tracing algorithm based on LACCs to identify colluders.
They also found an equivalence between an LACC and a binary separable code (SC),
and showed that binary frameproof codes (FPCs),
which  were widely considered as having no traceability for generic digital data,
actually have traceability in averaging attack. And then many results of LACCs and SCs were obtained \cite{Bl2,CFJLM1,CJM,GG}.
Jiang {\it et al.}  \cite{JCM} introduced the concept of a strongly separable code (SSC) and
gave the corresponding tracing algorithm to  resist averaging attack.
They also showed that a binary SSC has more codewords than a  binary FPC but has the same traceability as a binary FPC.
Recently, Gu {\it et al.} \cite{GVM} proposed binary secure codes with list decoding (SCLDs),
and proved that binary SCLDs have not only much more efficient traceability
than separable codes but also a much larger code rate than frameproof codes.
Finally, strongly identifiable parent property code for multimedia fingerprinting (SMIPPC)
was introduced to resist averaging attack.
The authors in \cite{JGC} also showed that a binary SMIPPC can be used to trace at least one colluder.

 For the above fingerprinting codes, we typically follow a two-stage paradigm:
first defining the code structure, then developing algorithms based on this framework.
Notably, once the algorithmic framework is established,
our priorities often shift toward optimizing code structures while neglecting algorithmic improvements,
a pattern that frequently results in low code rates.
This study breaks from conventional approaches by centering innovation on algorithm design,
and ultimately achieves the following technical breakthroughs.
\begin{itemize}

\item We propose a novel soft tracing algorithm to identify all colluders in averaging attacks.
Our analysis reveals three classes of binary fingerprinting codes compatible with this algorithm:
FPCs, SSCs, and SMIPPCs. Notably, binary SMIPPCs demonstrate equivalent traceability to
binary FPCs (or binary SSCs) while achieving higher code rates.
Furthermore, we establish that binary SMIPPCs outperform binary SCs (or binary SCLDs) in both traceability
and code rate.
In summary, based on the soft tracing algorithm,
we will improve code rate of codes that can be used to trace all colluders in averaging attack.

\item
Inspired by concatenated codes, we designed a two-stage soft tracing algorithm
to identify  all colluders in averaging attack.
The code that satisfy the conditions of this algorithm can be obtained by
concatenating a $q$-ary  SMIPPC with a binary SMIPPC.
It is worth noting that
our concatenated codes exhibit superior traceability to existing fingerprinting codes.
\end{itemize}

The rest of this paper is organized as follows.
In Section \ref{sec-pre}, we recap the averaging attack model in multimedia fingerprinting.
In Section \ref{sec-AACC}, we provide a  concept of an anti-averaging-collusion code (AACC),
and show that it can be used to identify all colluders in averaging attack.
In Section \ref{sec-STA} and Section \ref{sec-two-STA},
we propose two algorithms called soft tracing algorithm and two-stage soft tracing algorithm
to identify colluders and show their performances, respectively.
Conclusion is drawn in Section \ref{sec-conclu}.

\section{The Averaging Attack Model in Multimedia Fingerprinting}
\label{sec-pre}

In this section, we briefly review the averaging attack model in multimedia fingerprinting. 
The interested reader is referred to \cite{LTWWZ} for more details.

Spread-spectrum additive embedding is a widely employed robust embedding technique \cite{CKLS,PZ},
which is nearly capacity optimal when the host signal is available in detection \cite{CW,LTWWZ}.
Let ${\bf h}$ be the host multimedia signal and  $\{ {\bf u}_i \ | \ i \in \{1,2,\ldots,n\}\}$ be an orthonormal basis of noise-like signals.
We can choose appropriate $c_{i,j} \in \{0,1\}$, $ i \in \{1,2,\ldots,n\},  j \in \{1,2,\ldots,M\}$,
and obtain a family of watermarks
\begin{equation}
\label{eq-watermark}
 \{{\bf w}_j = ({\bf w}_j(1), {\bf w}_j(2), \ldots, {\bf w}_j(n))
 = \sum_{i=1}^{n}c_{i,j}{\bf u}_i\ | \   j \in \{1,2,\ldots,M\}\}.
\end{equation}
Obviously,  ${\bf w}_j$ can be represented uniquely by a vector (called codeword)
${\bf c}_j = (c_{1,j}, c_{2,j}, \ldots, c_{n,j})$. 
Then content with  the watermarks ${\bf w}_j$, i.e., ${\bf y}_j ={\bf h} + \alpha{\bf w}_j$, is assigned to the authorized user $U_j$,
where the parameter $\alpha \in \mathbb{R}^{+}$ is used to scale the watermarks to achieve the imperceptibility as well as to control the energy of the embedded watermark.

Without loss of generality, suppose that $U_{1}, U_{2}, \ldots, U_{t}$ are authorizes users, and amount a collusion attack.
In this process, each user  cannot manipulate the individual orthonormal signals, that is,
the underlying codeword needs to be taken and proceeded as a single entity,
but the users can carry on a linear collusion attack to generate a pirate copy from their watermarked contents,
so that the venture traced by the pirate copy can be attenuated.
For the averaging attack, 
one can extract
\begin{equation*}
 {\bf y} = \frac{1}{t}\sum\limits_{j=1}^{t}{\bf y}_{j} = \frac{\alpha}{t}\sum\limits_{j=1}^{t}{\bf w}_{j} + {\bf h}=
 \alpha\sum\limits_{j=1}^{t}\sum\limits_{i=1}^{n}\frac{c_{i,j}}{t}{\bf u}_{i} + {\bf h}
\end{equation*}
from the pirated content.

In colluder detection phase, by using the extracted vector ${\bf y}$,
we can compute ${\bf x}(i) = \langle \frac{{\bf y}-{\bf h}}{\alpha}, {\bf u}_{i}\rangle$,
where $ i \in \{1,2,\ldots,n\}$
and $\langle \frac{{\bf y}-{\bf h}}{\alpha}, {\bf u}_{i}\rangle$ is the inner product of $\frac{{\bf y}-{\bf h}}{\alpha}$ and  ${\bf u}_i$.
It is not difficult to check that
\begin{equation}\label{eq-generated-word-coord}
{\bf x}(i) = \frac{1}{t}\sum_{j=1}^{t}c_{i,j}= \frac{1}{t}\sum\limits_{j=1}^{t}{\bf c}_{j}(i)
\end{equation}
for any  $i \in \{1,2,\ldots,n\}$.
Let ${\bf x} = ( {\bf x}(1), {\bf x}(2), \ldots, {\bf x}(n))$. Then
\begin{equation}\label{eq-generated-word}
{\bf x} = \frac{1}{t}\sum\limits_{j=1}^{t}{\bf c}_{j}.
\end{equation}
We refer to the vector ${\bf x}$ as  the \textit{generated word} of
$\{{\bf c}_{1}, {\bf c}_{2}, \ldots, {\bf c}_{t}\}$ based on  averaging attack.
For convenience,
denote
\begin{equation}\label{equa-gene-word}
{\bf x} = {\sf AT}(\{{\bf c}_{1}, {\bf c}_{2}, \ldots, {\bf c}_{t}\}) =(\frac{a_1}{t_1}, \frac{a_2}{t_2}, \ldots, \frac{a_n}{t_n})
\end{equation}
where for any $i \in \{1,2,\ldots,n\}$, $a_i, t_i \in \mathbb{N}$ and
\begin{equation*}
\left\{\begin{array}{rl}
a_i = 0, t_i =1, \  &   \mbox{if}  \ {\bf x}(i) = 0,\\[2pt]
\gcd (a_i,t_i) =1,  &  \mbox{otherwise}.\\[2pt]
\end{array}
\right.
\end{equation*}

To identify all the colluders $U_{1}, U_{2}, \ldots, U_{t}$ by using the generated word ${\bf x}$,
we need to choose appropriate $c_{i,j}$,
i.e., construct codes with appropriate properties, and design corresponding tracing algorithms.

In this paper, we concentrate on the multimedia fingerprinting codes to resist averaging attacks,
and we will omit the word ``averaging attack" unless otherwise stated.
Based on the above discussion, there is a one-to-one mapping between an authorized user and the assigned codeword.
Thus we also make no difference between an authorized user and his/her corresponding codeword. 

\section{ Anti-Averaging-Collusion Code}
\label{sec-AACC}


Inspired the concept of codes with totally secure in  generic digital fingerprinting
(see e.g. \cite{BS}), in this section, we introduce a  concept of an anti-averaging-collusion code (AACC),
and demonstrate that the qualified AACCs could identify all the colluders by employing the feature of averaging attacks.

Let $n, M$ and $q$ be positive integers, and $Q = \{0, 1, \ldots, q-1\}$ an alphabet.
A set $\mathcal{C} = \{{\bf c}_1,{\bf c}_2,\ldots, {\bf c}_M\} \subseteq Q^n$ is called an $(n,M,q)$ \textit{code}
and each ${\bf c}_i$ is called a \textit{codeword}.
We also use the word ``\textit{binary}" if $q=2$.
Given an $(n,M,q)$ code, its incidence matrix is an $n \times M$ matrix on $Q$ where each column is a codeword in $\mathcal{C}$.
In the sequel, we make no difference between an $(n,M,q)$ code and its incidence matrix. 

For any code $\mathcal{C} \subseteq Q^n$ and  $\mathcal{C}' \subseteq \mathcal{C}$, let  $\mathcal{C}'(i)$
denote the set of $i$-th components of codewords in $\mathcal{C}'$,  $ i \in \{1,2,\ldots,n\}$.
As in \cite{HLLT},
we define the 
\textit{descendant code} of $\mathcal{C}'$ as
\begin{equation}\label{def-desc}
  {\sf desc}(\mathcal{C}')=\mathcal{C}'(1) \times \mathcal{C}'(2) \times \cdots \times  \mathcal{C}'(n).
\end{equation}

\begin{example} \label{Examdes}
Consider the following $(4,5,2)$ code $\mathcal{C}$, and $\mathcal{C}' = \{{\bf c}_1, {\bf c}_2, {\bf c}_3\} \subseteq \mathcal{C}$.
\begin{eqnarray*}
\mathcal{C}&=& \left(
  \begin{array}{cccccc}
     {\bf c}_1&   {\bf c}_2 &  {\bf c}_3& {\bf c}_4& {\bf c}_5\\ \hline
     0 & 0 & 0  & 1 & 0\\
     1 & 1 & 0  & 0 & 0\\
     1 & 0 & 1  & 0 & 0\\
     0 & 0 & 1  & 1 & 0
  \end{array}
\right)
\end{eqnarray*}

According to Formula \eqref{def-desc}, we have that
\begin{equation*}
  \mathcal{C}'(1) = \{0\}, \ \mathcal{C}'(2) = \{0,1\}, \ \mathcal{C}'(3) = \{0, 1\}, \ \mathcal{C}'(4) = \{0,1\},
\end{equation*}
and
\begin{equation}\label{equ-exam-desc}
  {\sf desc}(\mathcal{C}') = \{0\} \times \{0,1\} \times \{0,1\} \times \{0,1\} =
  \{ {\bf c}_1, {\bf c}_2, {\bf c}_3, {\bf c}_5, (0,0,0,1)^{T}, (0,0,1,0)^{T}, (0,1,0,1)^{T}, (0,1,1,1)^{T}\}.
\end{equation}
\end{example}

Similar to the codes with totally secure  in  generic digital fingerprinting (see e.g. \cite{BS}),
we define the concept of an anti-averaging-collusion code for multimedia fingerprinting as follows.
\begin{definition}\label{def_AACC}
A binary code with a tracing algorithm $\mathcal{\phi}$ is called
\textit{$t$-anti-averaging-collusion code}, or $t$-AACC,
if $\mathcal{\phi}({\bf x})=\mathcal{C}_0$ holds for any $\mathcal{C}_0 \subseteq \mathcal{C}$ with $1 \leq |\mathcal{C}_0|\le t$
and ${\bf x}={\sf AT}(\mathcal{C}_0)$.
\end{definition}

According to  Definition \ref{def_AACC},
an AACC contains two important factors, i.e., a binary code and a tracing algorithm.
In contrast to prior research practices,
we propose the algorithm first and subsequently seek codes that fulfill its requirements
in the following two sections.
Through comparison, we find that this method yields codes with significantly higher code rates.

\begin{example} \label{ExamAACC}
We claim that the code $\mathcal{C}$  in Example \ref{Examdes}
with an algorithm $\mathcal{\phi}$ which will be described later is a $3$-AACC.
According to Definition \ref{def_AACC},
we need to show that $\mathcal{\phi}({\bf x})=\mathcal{C}_0$ holds for any $\mathcal{C}_0 \subseteq \mathcal{C}$ with $1 \leq |\mathcal{C}_0|\le 3$
and ${\bf x}={\sf AT}(\mathcal{C}_0)$.
For instance,  $\mathcal{C}_0 = \{{\bf c}_1$, ${\bf c}_2,{\bf c}_3\}$, then
${\bf x} = (0, \frac{2}{3},\frac{2}{3},\frac{1}{3})$
comes from Formulas \eqref{eq-generated-word-coord} and \eqref{eq-generated-word}.
Next, we describe the idea of the tracing algorithm.
\begin{itemize}
\item[1)] Determine the exact number of colluders $|\mathcal{C}_0|$ by using  the generated word ${\bf x}$.
  Observe that  ${\bf x}(3)=\frac{1}{3}$ and the fact that  $|\mathcal{C}_0|\leq 3$,
  we have that $|\mathcal{C}_0|=3$.
\item[2)] Trace colluders.
    \begin{itemize}
        \item[2-1)] The first iteration.
        \begin{itemize}
      \item Compute ${\sf desc}(\mathcal{C}_0)$ by ${\bf x}$.
  Let ${\bf R} = {\bf R}(1)\times {\bf R}(2)\times {\bf R}(3) \times {\bf R}(4)$, where
\begin{equation*}
 {\bf R}(i) =
\left\{\begin{array}{rl}
\{0\},   &   \mbox{if}  \ \ {\bf x}(i) =0,\\[2pt]
\{0,1\},    &   \mbox{if}  \ \ 0 < {\bf x}(i) < 1,\\[2pt]
\{1\},    &   \mbox{if}  \ \  {\bf x}(i) = 1.\\[2pt]
\end{array}
\right.
\end{equation*}
Then ${\bf R} = \{0\} \times \{0,1\}  \times \{0,1\}  \times \{0,1\}$.
It is not difficult to check that, for any $ i \in \{1,2,\ldots,n\}$, ${\bf R}(i) = \mathcal{C}_0(i)$ and
${\bf R}(i)$ reveals the elements of $i$th coordinate of all the colluders.
Hence ${\sf desc}(\mathcal{C}_0) = {\bf R}$.
  \item Delete all the innocent users who can not be framed by the colluder set $\mathcal{C}_0$.
     That is, for any $ i \in \{1,2,\ldots,n\}$ with ${\bf R}(i) = \{ 0 \}$ or ${\bf R}(i) = \{ 1 \}$,
     we can delete the codeword ${\bf c} \in \mathcal{C}$ such that ${\bf c}(i) \notin {\bf R}(i)$.
     Actually, we have deleted the codewords ${\bf c} \in \mathcal{C}$ such that ${\bf c} \notin {\sf desc}(\mathcal{C}_0)$,
     which implies that ${\bf c}$ is not a colluder since the colluder set $\mathcal{C}_0 \subseteq{\sf desc}(\mathcal{C}_0)$.
     So, the set of the rest codewords is in fact ${\sf desc}(\mathcal{C}_0)\cap \mathcal{C}$.
     According to the above discussion, ${\sf desc}(\mathcal{C}_0)\cap \mathcal{C} = \{{\bf c}_1, {\bf c}_2, {\bf c}_3, {\bf c}_5\}$.
  \item Determine colluders.
  For any $ i \in \{1,2,\ldots,n\}$ with ${\bf R}(i) = \{0,1\}$, find out one  codeword
     ${\bf c} \in {\sf desc}(\mathcal{C}_0)\cap \mathcal{C}$, such that
     \begin{equation}\label{eq-unique}
     {\bf c}(i) \neq {\bf c}'(i) \ \textrm{for any} \ {\bf c}' \in ({\sf desc}(\mathcal{C}_0)\cap \mathcal{C}) \setminus \{{\bf c}\}.
     \end{equation}
     Then ${\bf c}$ must be a colluder since the symbol ${\bf c}(i)$ in ${\bf R}(i)$ is certainly contributed by ${\bf c}$
     from the uniqueness in \eqref{eq-unique}.
     So ${\bf c}_3$ is identified  by using the condition  ${\bf R}(4)=\{0,1\}$ in this step.

        \end{itemize}
  \item[2-2)] The second iteration.
  \begin{itemize}
    \item Update generated word ${\bf x}' =\frac{|\mathcal{C}_0|}{|\mathcal{C}_0|-1}({\bf x} -
     \frac{1}{|\mathcal{C}_0|}{\bf c}_3) = (0, 1,\frac{1}{2},0).$
     We remark that ${\bf x}'$ is exactly the generated word by $\mathcal{C}_0 \setminus \{{\bf c}_3\}$,
     i.e., ${\bf x}' = {\sf AT}(\mathcal{C}_0 \setminus \{{\bf c}_3\})$.
     Such a condition is the key to the algorithm's ability to track all users.
      \item Compute ${\sf desc}(\mathcal{C}_0 \setminus \{{\bf c}_3\})$ by ${\bf x}'$.
      Similarly, we can obtain that  ${\bf R} = \{0\} \times \{1\}  \times \{0,1\}  \times \{0\}$.
      That is ${\sf desc}(\mathcal{C}_0 \setminus \{{\bf c}_3\})={\bf R}$.
      \item Delete all the innocent users.
      Similarly, we can obtain that
      ${\sf desc}(\mathcal{C}_0 \setminus \{{\bf c}_3\})\cap \mathcal{C} = \{{\bf c}_1, {\bf c}_2\}$.
      \item Determine colluders.
      Similarly, ${\bf c}_1$ and ${\bf c}_2$ are identified  by using  the condition ${\bf R}(3)=\{0,1\}$ in this step.
  \end{itemize}

  \end{itemize}
Following the two iterative phases,
  we know that ${\bf c}_1$, ${\bf c}_2$ and ${\bf c}_3$ are colluders.
  i.e, $\mathcal{\phi}({\bf x})=\mathcal{C}_0$.
\end{itemize}

For any other subset $\mathcal{C}_0\subseteq \mathcal{C}$ with $1 \leq |\mathcal{C}_0| \leq 3$,
we could use a similar way to check that $\mathcal{C}_0$ satisfies the above condition.
So $\mathcal{C}$ is a $3$-AACC.

\end{example}

One can immediately to derive the following result according to  Definition \ref{def_AACC}.
\begin{theorem}\label{th-AACC-tracing}
Any $t$-AACC  can be applied to identify all colluders under the assumption that
the number of colluders in the averaging attack is at most $t$.
\end{theorem}

\section{Soft Tracing Algorithm}
\label{sec-STA}

\subsection{Known Codes and Corresponding Tracing Algorithms}

Firstly, we list several known fingerprinting codes,
and the computational complexities of their corresponding tracing algorithms.
In the literature, it is known that these codes are equipped with decoding algorithms to trace back to all colluders.

\begin{definition} \rm(\cite{BS,CM,GVM,JCM})
\label{def-LACC-FPC-SSC}
Let $\mathcal{C}$ be an $(n,M,q)$ code and, $t$ and $L$ are positive integers with $2 \leq t \leq L\leq M$.
\begin{enumerate}

\item $\mathcal{C}$ is a \textit{$\overline{t}$-separable code}, or $\overline{t}$-SC$(n, M, q)$,
if for any distinct $\mathcal{C}_{1}, \mathcal{C}_{2}\subseteq \mathcal{C}$ with
 $1 \leq |\mathcal{C}_1|, |\mathcal{C}_2| \le t$,
 we have ${\sf desc}(\mathcal{C}_{1}) \neq $ ${\sf desc}(\mathcal{C}_{2})$.
\item  $\mathcal{C}$ is a \textit{$\overline{t}$-secure code with list decoding},
or $\overline{t}$-{\rm SCLD}$(n,M,q;L)$,
if for any distinct $\mathcal{C}_{1}, \mathcal{C}_{2}\subseteq \mathcal{C}$ with
 $1 \leq |\mathcal{C}_1|, |\mathcal{C}_2| \le t$,
we have  ${\sf desc}(\mathcal{C}_{1}) \neq $ ${\sf desc}(\mathcal{C}_{2})$ and
$|{\sf desc}(\mathcal{C}_{1}) \cap \mathcal{C} | \leq L$.
  \item $\mathcal{C}$ is a \textit{strongly $\overline{t}$-separable code},
  or $\overline{t}$-{\rm SSC}$(n,M,q)$, if for any
$\mathcal{C}_0 \subseteq \mathcal{C}$ with $1 \le |\mathcal{C}_0| \leq t$,
we have $\cap_{\mathcal{S} \in \mathcal{P}(\mathcal{C}_0)}\mathcal{S} = \mathcal{C}_0$,
where $\mathcal{P}(\mathcal{C}_0) = \{ \mathcal{S} \subseteq \mathcal{C} \ | \ {\sf desc}(\mathcal{S}) = {\sf desc}(\mathcal{C}_0) \}$.
\item $\mathcal{C}$ is a \textit{$t$-frameproof code}, or $t$-FPC$(n,M,q)$, if
for any $\mathcal{C}_0 \subseteq \mathcal{C}$ with $1 \leq |\mathcal{C}_0| \le t$,
we have ${\sf desc}(\mathcal{C}_0) \cap \mathcal{C} = \mathcal{C}_0$.
\end{enumerate}
\end{definition}

\begin{proposition}
\label{pro-knownAACCs}
\begin{enumerate}
 \item Any $\overline{t}$-SC$(n,M,2)$ with its corresponding tracing algorithm in \cite{CM}
  is a $t$-AACC. The computational complexity of the tracing algorithm is $O(nM^t)$.
  \item Any $\overline{t}$-{\rm SCLD}$(n,M,2;L)$, with its corresponding tracing algorithm in \cite{GVM}
  is a $t$-AACC. The computational complexity of the tracing algorithm is $O(\max\{nM,nL^t\})$.
   \item Any $\overline{t}$-SSC$(n,M,2)$ with its corresponding tracing algorithm in \cite{JCM}
  is a $t$-AACC. The computational complexity of the tracing algorithm is $O(nM)$.
  \item Any $t$-FPC$(n,M,2)$ with its corresponding tracing algorithm in \cite{CM}
  is a $t$-AACC. The computational complexity of the tracing algorithm is $O(nM)$.

\end{enumerate}
\end{proposition}

\begin{IEEEproof} We only prove the first statement,
since the other three statements can be derived by a similar method.
Suppose that $\mathcal{C}$ is a $\overline{t}$-SC$(n,M,2)$.
For any $\mathcal{C}_0\subseteq \mathcal{C}$ such that $1 \leq | \mathcal{C}_0 | \leq t$,
let $\mathcal{C}_0$ be the set of all the colluders, and ${\bf x}= {\sf AT}(\mathcal{C}_0)$.
Let ${\bf R} = {\bf R}(1)\times {\bf R}(2)\times \cdots \times {\bf R}(n)$,
where
\begin{equation*}
 {\bf R}(i) =
\left\{\begin{array}{rl}
\{0\},   &   \mbox{if}  \ \ {\bf x}(i) =0,\\[2pt]
\{0,1\},    &   \mbox{if}  \ \ 0 < {\bf x}(i) < 1,\\[2pt]
\{1\},    &   \mbox{if}  \ \  {\bf x}(i) = 1.\\[2pt]
\end{array}
\right.
\end{equation*}
for any $i \in \{1,2,\ldots,n\}$.
One can directly check that ${\bf R}(i) = \mathcal{C}_0(i)$ holds for any $ i \in \{1,2,\ldots,n\}$.
Thus ${\bf R} = {\sf desc}(\mathcal{C}_0)$.
We now compute the descendent code of each subset with the size at most $t$ of $\mathcal{C}$,
and find out the subset $\mathcal{C}'$ such that ${\sf desc}(\mathcal{C}') = {\bf R}$.
Hence ${\sf desc}(\mathcal{C}') ={\sf desc}(\mathcal{C}_0)$.
According to the definition of an SC,
${\sf desc}(\mathcal{C}_1) \neq {\sf desc}(\mathcal{C}_2)$
for any distinct subsets $\mathcal{C}_1, \mathcal{C}_2 \subseteq \mathcal{C}$
with $1 \leq |\mathcal{C}_1|, |\mathcal{C}_2| \leq t$.
Thus $\mathcal{C}' = \mathcal{C}_0$.
That is, $\phi({\bf x})=\mathcal{C}_0$.
So  $\mathcal{C}$ is an AACC.
\end{IEEEproof}

Together with the results in \cite{GVM} and \cite{JCM},
we summarize the relationships among the above fingerprinting codes in Table \ref{tab-rela-known-codes}.
{\begin{table}[h]
\begin{center}
\caption{Relationships among AACC and different types of known fingerprinting codes}
\label{tab-rela-known-codes}
   \begin{tabular}{ccccccccc}
     &  & $t$-FPC$(n,M,q)$ &  $ \Longrightarrow $ & $\overline{t}$-SSC$(n,M,q)$  \\
     & & $ \Downarrow $  &  & $ \Downarrow $ \\
     & & SCLD$(n,M,q; L)$  & $ \Longrightarrow $ & $\overline{t}$-SC$(n,M,q)$
     &
     $ \overset{q=2}{\Longrightarrow}  $ & AACC$(n,M,2)$
  \end{tabular}
\end{center}
\end{table}}
\begin{remark}According to  Formula \eqref{eq-watermark},
only binary codes can be used to construct fingerprints resistant
to averaging attacks on multimedia contents under the spread-spectrum additive embedding.
However, directly constructing binary codes is a very difficult task.
The common approach we use is to first construct $q$-ary codes and then convert them
into binary codes through specialized methods.
For example, the simplest method involves concatenating the $q$-ary code with unit vectors,
which yields a binary code that preserves the properties of the original code.
Therefore, investigating $q$-ary codes is also an interesting work.
\end{remark}

We now give computational complexities of corresponding tracing algorithms for
different types of fingerprinting codes in Table \ref{tab-comp-konwn-codes}.

{\begin{table*}[h]
\center
\caption{Computational complexities of tracing algorithms of binary fingerprinting codes}
\label{tab-comp-konwn-codes}
  \begin{tabular}{|c|c|c|c|c|c|c|c|c|}
\hline
 & FPC &   SSC  &   SC &  SCLD   \\ \hline
\tabincell{c}{Complexity}
    &   $O(nM)$
    &   $O(nM)$
    &     $O(nM^t)$
    &   $O(\max\{nM,nL^t\})$ \\  \hline

Reference & \cite{Bla,CM} &  \cite{JCM}  &   \cite{CJM,CM} &   \cite{GVM}\\ \hline
\end{tabular}
\end{table*}

Finally, we list code rates of different types of codes.
Let $M_{\rm FPC}(t, n, q)$, $M_{\rm SSC}(\overline{t}, n, q)$, $M_{\rm SC}(\overline{t}, n, q)$,
$M_{\rm SCLD}(\overline{t}, n, q)$
denote the largest cardinality of a $q$-ary $t$-FPC, $\overline{t}$-SSC, $\overline{t}$-SC,
$\overline{t}$-SCLD of  length $n$, respectively.
Then we can denote their largest asymptotic code rates as
\begin{eqnarray*}
&&R_{\rm FPC}(t,n) = \limsup_{q \rightarrow \infty}\frac{\log_q M_{\rm FPC}(t, n, q)}{n},\\
&&R_{\rm SSC}(\overline{t},n)
  =\limsup_{q \rightarrow \infty}\frac{\log_q M_{\rm SSC}(\overline{t}, n, q)}{n},\\
&&R_{\rm SC}(\overline{t},n) =\limsup_{q \rightarrow \infty}\frac{\log_q M_{\rm SC}(\overline{t}, n, q)}{n},\\
&&R_{\rm SCLD}(\overline{t},n,L)
  =\limsup_{q \rightarrow \infty}\frac{\log_q M_{\rm SCLD}(\overline{t}, n, q;L)}{n}.
\end{eqnarray*}

 We list the state-of-the-art bounds about FPCs, SSCs, SCs, SCLDs
in Table \ref{tab-konwn-rates}.

{\begin{table*}[h]
\center
\caption{Code rates of different types of $q$-ary codes when $q\rightarrow \infty$}
\label{tab-konwn-rates}
  \begin{tabular}{|c|c|c|c|c|c|c|c|}
\hline
 & $R_{\rm FPC}(t,n)$
 &  $R_{\rm SSC}(\overline{t},n)$, $R_{\rm SC}(\overline{t},n)$ & $R_{\rm SCLD}(\overline{t},n;L)$ \\ \hline
\tabincell{c}{Code Rate}
    &   $=\frac{\lceil n/t \rceil}{n}$
    &  \tabincell{c}
    {$\leq
      \left\{\begin{array}{rl}
         \frac{\lceil 2n/3 \rceil}{n}, \  &   \mbox{if}  \ t=2,\\[2pt]
         \frac{\lceil n/(t-1) \rceil}{n},  &  \mbox{if}  \ t>2.\\[2pt]
             \end{array}
      \right.$}
    &\tabincell{c}
    {$\leq
      \left\{\begin{array}{rl}
         \frac{\lceil 2n/3 \rceil}{n}, \  &   \mbox{if}  \ t=2, L\geq 3,\\[2pt]
         \frac{\lceil n/(t-1) \rceil}{n},  &  \mbox{if}  \ t>2, L\geq t+1.\\[2pt]
             \end{array}
      \right.$}

 \\  \hline
Reference & \cite{Bla} & \cite{Bl2,JCM}  & \cite{GVM} \\ \hline
\end{tabular}
\end{table*}

\subsection{Soft Tracing Algorithm}

Next, we will introduce a new tracing algorithm called a soft tracing algorithm (Algorithm \ref{alg-soft-TA}).
Here, we only illustrate the ideas of our algorithms,
and we will establish the conditions for the algorithms' validity later.

\begin{itemize}
\item[1)] Algorithm \ref{alg-desc}:
 Suppose that $\mathcal{C}$ is an $(n, M, 2)$ code,
 $\mathcal{C}_0 \subseteq \mathcal{C}$ with $|\mathcal{C}_0|=t_0$,
 and ${\bf x}={\sf AT}(\mathcal{C}_0)$ being of the form in Formula (\ref{equa-gene-word}).
 When the input is  ${\bf x}$,
 we expect the algorithm to output the descendent code of $\mathcal{C}_0$, i.e., ${\bf R} = {\sf desc}(\mathcal{C}_0)$,
 where ${\bf R}$ is the output of the algorithm.


\item[2)] Algorithm \ref{alg-q-find-inter}:
Suppose that $\mathcal{C}$ is an $(n, M, q)$ code,
 and $\mathcal{C}_0 \subseteq \mathcal{C}$ with $|\mathcal{C}_0|=t_0$.
  When the input is ${\bf R}= {\sf desc}(\mathcal{C}_0)$,
  we expect the algorithm to output an index set of a subset of $\mathcal{C}_0$.
We  starts with the entire group as the suspicious set, i.e.,  $X=\{1,2,\ldots,M\}$ in Line $1$.
  Then we delete the index $j$ such that ${\bf c}_j \notin {\sf desc}(\mathcal{C}_0)$,
  and obtain the suspicious set $X=\{X(1),X(2), \ldots, X(|X|)\}$ in Line $9$.
Finally, by using Lines 9-32 of the algorithm,
  we will find out an index set $U \subseteq X$ satisfying that for any $h \in U$,
  there exists an integer $i\in \{1, 2, \ldots, n\}$ such that
   ${\bf c}_h(i)$ is a unique element in ${\bf R}(i)$,
i.e., ${\bf c}_h(i) \neq {\bf c}_{h'}(i)$ for any $h' \in X$.
In a word,  we expect that $\{{\bf c}_h \ | \ h \in U\}$  is a subset of  $\mathcal{C}_0$,
where $U$ is the output of the algorithm.
\item[3)]Algorithm \ref{alg-soft-TA}:
Suppose $\mathcal{C}$ is an $(n,M,2)$ code, $\mathcal{C}_0\subseteq \mathcal{C}$
and ${\bf x}={\sf AT}(\mathcal{C}_0)$ being of the form in Formula (\ref{equa-gene-word}).
When the input is ${\bf x}$,
we expect the algorithm to output $\mathcal{C}_0$.
We first determine the exact number of colluders,
i.e., we expect $t_0 = \max\{t_1, t_2, \ldots, t_{n}\}$ is equal to $|\mathcal{C}_0|$.
Next, the algorithm will enter the while loop.
During the first iteration,
Algorithm \ref{alg-desc} outputs the descendant code of $\mathcal{C}_0$,
i.e., ${\bf R}={\sf desc}(\mathcal{C}_0)$
where ${\bf R}$ is the set in Line $6$ of the algorithm.
Then Algorithm \ref{alg-q-find-inter} outputs an index set of
a subset $\mathcal{C}_1$ of $\mathcal{C}_0$.
In the second iteration, the generated word is updated, i.e.,
${\bf x}$ is in fact the generated word of $\mathcal{C}_0 \setminus \mathcal{C}_1$.
Similarly, Algorithm \ref{alg-q-find-inter} outputs an index set of a subset
$\mathcal{C}_2 \subseteq \mathcal{C}_0 \setminus \mathcal{C}_1$.
Repeat this process.
In the last iteration, Algorithm \ref{alg-q-find-inter} outputs an index set of a subset
$\mathcal{C}_s \subseteq \mathcal{C}_0 \setminus \cup_{i=1}^{s-1}\mathcal{C}_i$.
Now we expect that $\mathcal{C}_0 = \cup_{i=1}^{s}\mathcal{C}_i$.
In a word, we expect $\{{\bf c}_j \ | \ j \in U\}$ is equal to $\mathcal{C}_0$,
where $U$ is the output of the algorithm.
\end{itemize}

\IncMargin{2em}
\begin{algorithm}[H]
\caption{ {\tt DescAlg}}
\label{alg-desc}
\SetKwData{Left}{left}\SetKwData{This}{this}\SetKwData{Up}{up}
\SetKwFunction{Union}{Union}\SetKwFunction{FindCompress}{FindCompress}
\SetKwInOut{Input}{input}\SetKwInOut{Output}{output}

\Input {${\bf x} = ({\bf x}(1), {\bf x}(2), \ldots, {\bf x}(n))$}

\For { $i=1$  {\bf to }  $n$ }
     { \If {${\bf x}(i) = 0$}
           {  ${\bf R}(i) = \{0\}$\;}
       \Else
        {  \If {${\bf x}(i) = 1$}
               {  ${\bf R}(i) = \{1\}$\;}
           \Else
               {  ${\bf R}(i)= \{0, 1\}$\;}
        }
    }

\Output {${\bf R} = {\bf R}(1) \times {\bf R}(2) \times \ldots \times {\bf R}(n)$}

\end{algorithm}

\IncMargin{2em}
\begin{algorithm}[H]
\caption{ {\tt FindInterAlg}}
\label{alg-q-find-inter}
\SetKwData{Left}{left}\SetKwData{This}{this}\SetKwData{Up}{up}
\SetKwFunction{Union}{Union}\SetKwFunction{FindCompress}{FindCompress}
\SetKwInOut{Input}{input}\SetKwInOut{Output}{output}

\Input {${\bf R} = {\bf R}(1) \times {\bf R}(2) \times \ldots \times {\bf R}(n)$}

$X = \{1,2, \ldots, M\}$;

\For { $i=1$ {\bf to} $n$}
     {

         \For { $j=1$ {\bf to} $M$}
              {   \If { ${\bf c}_j(i)\notin {\bf R}(i)$  }
                      { $X = X \setminus \{j\}$;
                      }
              }
     }

Denote $X=\{X(1),X(2), \ldots, X(|X|)\}$;

$U=\emptyset$\;

\For { $i=1$ {\bf to} $n$}
     {
        \For { $k=0$ {\bf to} $q-1$}
             { $r_{i,k} = 0$;
             }

        \For { $j=1$ {\bf to} $|X|$}
             {  $r_{i,{\bf c}_{X(j)}(i)}= r_{i,{\bf c}_{X(j)}(i)}+1$;
             }

        $Y={\bf R}(i)$;

         \For { $k=0$ {\bf to} $q-1$}
        {
           \If {$r_{i,k} \neq 1$}
           {$Y=Y\setminus \{k\}$\;}
        }

        Denote $Y=\{Y(1),Y(2), \ldots, Y(|Y|)\}$;

        \For { $k=1$ {\bf to} $|Y|$}
        {
               {  \For { $j=1$ {\bf to} $|X|$}
                       {  \If {${\bf c}_{X(j)}(i)=Y(k)$ }
                              { $U=U \cup \{X(j)\}$\;
                              }
                       }
               }
        }
     }

\Output {the index set $U$}\label{End}

\end{algorithm}
\vskip 1cm

\IncMargin{2em}
\begin{algorithm}[H]
\caption{ {\tt Soft Tracing Algorithm}}
\label{alg-soft-TA}
\SetKwData{Left}{left}\SetKwData{This}{this}\SetKwData{Up}{up}
\SetKwFunction{Union}{Union}\SetKwFunction{FindCompress}{FindCompress}
\SetKwInOut{Input}{input}\SetKwInOut{Output}{output}

\Input {${\bf x} = ({\bf x}(1), {\bf x}(2), \ldots, {\bf x}(n))$}

Denote $t_0 = \max\{t_1, t_2, \ldots, t_{n}\}$\;

$U=\emptyset$;

$Flag=True$;

\While {$Flag$ and $|U| < t_0$}
    {   ${\bf x}=\frac{t_0{\bf x}-\sum_{j\in U}{\bf c}_{j}}{t_0-|U|}$\;

        Execute Algorithm \ref{alg-desc} with the input ${\bf x}$, and the output is ${\bf R}$\;

        Execute Algorithm \ref{alg-q-find-inter} with the input ${\bf R}$, and the output is $U'$\;

        \If {$U'=\emptyset$}
            {   $Flag = False$ \;
            }

        \Else
            {$U=U\cup U'$\;}
         }

\If {$|U| \neq t_0$}
    {
       \Output {       This code does not satisfy the conditions of the algorithm.  }
    }
\Else
    { \Output {the index set $U$}}

\end{algorithm}

\vskip 1cm

Now we will characterize the properties required for codes to be applicable to the soft tracing algorithm.
Suppose that  $\mathcal{C}$  is an $(n,M,q)$ code.
Then we say that the code $\mathcal{C}$ has {\it $t$-uniqueness descendant code},
if for any subcode $\mathcal{C}_0\subseteq \mathcal{C}$ with $1 \leq |\mathcal{C}_0|\le t$,
there exist a codeword ${\bf c} \in \mathcal{C}_0$ and an integer $i \in \{1,2,\ldots,n\}$
such that ${\bf c}(i) \neq {\bf c}'(i)$
for any ${\bf c}' \in ({\sf desc}(\mathcal{C}_0) \cap \mathcal{C}) \setminus \{{\bf c}\}$.

\begin{theorem}\label{th-soft-TA}
Suppose that an $(n,M,2)$ code $\mathcal{C}$ has $t$-uniqueness descendant code.
Then under the assumption that the number of colluders in the averaging attack is at most $t$,
the code $\mathcal{C}$ can be applied to identify all colluders by using the soft tracing algorithm (Algorithm \ref{alg-soft-TA}), and the computational complexity is $O(tnM)$.

\end{theorem}

We will establish the above statement.
In fact, the idea of Algorithm \ref{alg-soft-TA} is the same as the algorithm in Example \ref{ExamAACC},
which implies that we first need to know the exact number of the colluders.
It is worth noting that the property of ``uniqueness" of ${\bf c}$
is very useful to prove Theorem \ref{th-soft-TA}.
Precisely, the property of ``uniqueness"  is not only useful for
determining  the exact number of the colluders,
but also useful for identifying the colluders.

We first determine the exact number of the colluders.
Suppose that  $\mathcal{C}$ is an $(n,M,2)$ code having $t$-uniqueness descendant code,
and $\mathcal{C}_0 = \{ {\bf c}_1, {\bf c}_2, \ldots, {\bf c}_{t_0}\} \subseteq \mathcal{C}$
is exactly the set of all the colluders, where $t_0 \leq t$. Let
${\bf x} = {\sf AT}(\mathcal{C}_0)$ be of the form in Formula (\ref{equa-gene-word}).
According to Formula \eqref{eq-generated-word-coord},
we know that ${\bf x}(i) = \frac{1}{t_0} \sum_{j = 1}^{t_0} {\bf c}_j(i)$ for any $i \in \{1,2,\ldots,n\}$.
Since $\mathcal{C}$ has $t$-uniqueness descendant code,
there exists a codeword ${\bf c} \in \mathcal{C}_0$ and an integer $i \in \{1,2,\ldots,n\}$
such that ${\bf c}(i) \neq {\bf c}'(i)$
for any ${\bf c}' \in ({\sf desc}(\mathcal{C}_0) \cap \mathcal{C}) \setminus \{{\bf c}\}$.
This implies that ${\bf c}(i) \neq {\bf c}'(i)$ for any ${\bf c}' \in \mathcal{C}_0\setminus \{{\bf c}\}$
as $\mathcal{C}_0\setminus \{{\bf c}\}  \subseteq ({\sf desc}(\mathcal{C}_0) \cap \mathcal{C})\setminus \{{\bf c}\}$.
\begin{itemize}
  \item If ${\bf c}(i) = 1$,  we have ${\bf c}'(i) = 0$ for any ${\bf c}' \in \mathcal{C}_0 \setminus \{{\bf c}\}$.
     Then ${\bf x}(i) = \frac{1}{t_0}$.
  \item If ${\bf c}(i) = 0$,  we have ${\bf c}'(i) = 1$ for any ${\bf c}' \in \mathcal{C}_0 \setminus \{{\bf c}\}$.
     Then ${\bf x}(i) = \frac{t_0 -1}{t_0}$.
\end{itemize}
So the denominator of ${\bf x}(i)$ is exactly equal to $t_0$, which is in fact the maximum number of $t_1, t_2, \ldots, t_n$.

\begin{proposition}
\label{propo-colludersNB}
Suppose that  $\mathcal{C}$ is an $(n, M, q)$ code with $t$-uniqueness descendant code,
$\mathcal{C}_0\subseteq \mathcal{C}$ with $1 \leq |\mathcal{C}_0|\le t$,
and  ${\bf x} = {\sf AT}(\mathcal{C}_0)$ being of the form in Formula (\ref{equa-gene-word}).
Then $|\mathcal{C}_0| = \max\{t_i  \ | \  i \in \{1,2,\ldots,n\}\}$.
\end{proposition}

According to Proposition \ref{propo-colludersNB},
the exact number of colluders can be determined by the observed vector ${\bf x}$
if a code has $t$-uniqueness descendant code.

\begin{proposition}
\label{propo-desc}
Suppose that $\mathcal{C}$ is an $(n,M,2)$ code, $\mathcal{C}_0\subseteq \mathcal{C}$,
and ${\bf x} = {\sf AT}(\mathcal{C}_0)$ being of the form in Formula (\ref{equa-gene-word}).
Then the output ${\bf R}$ of Algorithm \ref{alg-desc} is equal to ${\sf desc}(\mathcal{C}_0)$
if the input is ${\bf x}$.
The computational complexity is $O(n)$.
\end{proposition}

\begin{proof} We can directly check that  for any $i \in \{1,2,\ldots,n\}$,
the following conditions hold by Formula \eqref{eq-generated-word-coord} and Algorithm  \ref{alg-desc}.
  \begin{itemize}
    \item  ${\bf R}(i) = \{0\}$ if and only if ${\bf c}'(i) = 0$ for any ${\bf c}' \in \mathcal{C}_0$.
    \item  ${\bf R}(i) = \{1\}$ if and only if ${\bf c}'(i) = 1$ for any ${\bf c}' \in \mathcal{C}_0$.
    \item  ${\bf R}(i) = \{0,1\}$ if and only if there exist ${\bf c}', {\bf c}'' \in \mathcal{C}_0$
         such that ${\bf c}'(i) = 0$ and ${\bf c}''(i) = 1$.
  \end{itemize}
  According to the above discussions, we know that ${\bf R} = {\sf desc}(\mathcal{C}_0)$.
\end{proof}

In order to prove Theorem \ref{th-soft-TA}, the following lemma is needed.

\begin{lemma}
\label{lem-find-inter}
Suppose that $\mathcal{C}$ is an $(n, M, q)$ code with $t$-uniqueness descendant code,
and $\mathcal{C}_0 \subseteq \mathcal{C}$ with $1 \leq |\mathcal{C}_0|\leq t$.
Then  Algorithm \ref{alg-q-find-inter} outputs an index set of a non-empty subset of $\mathcal{C}_0$
if the input is ${\sf desc}(\mathcal{C}_0)$,
that is $U \neq \emptyset$ and $\{{\bf c}_j \ | \ j \in U\} \subseteq \mathcal{C}_0$,
where $U$ is the output of Algorithm \ref{alg-q-find-inter}.
In addition, the computational complexity is $O(\min\{t, q\}nM)$.
\end{lemma}
\begin{IEEEproof} Firstly, consider the computational complexity.
Since the input is
${\sf desc}(\mathcal{C}_0) = \mathcal{C}_0(1)\times \mathcal{C}_0(2) \times \cdots \times \mathcal{C}_0(n)$
we have ${\bf R}(i) = \mathcal{C}_0(i)$ for any $i \in \{1,2,\ldots,n\}$.
This implies that $|{\bf R}(i)| = |\mathcal{C}_0(i)| \leq \min\{t,q\}$.
In addition, according to Line $18$, we have $Y={\bf R}(i)$ which implies that  $|Y| =|{\bf R}(i)| \leq \min\{t,q\}$.
Therefore, the computational complexity of Algorithm \ref{alg-q-find-inter} is $O(\min\{t,q\}nM)$.

Since ${\bf R} = {\sf desc}(\mathcal{C}_0)$ and $X$ in Line $1$ of Algorithm \ref{alg-q-find-inter}
can be regarded as the set of subscripts of the codewords in $\mathcal{C}$,
Lines $1-9$ of Algorithm \ref{alg-q-find-inter} is to find out all the codewords in ${\sf desc}(\mathcal{C}_0) \cap \mathcal{C}$,
i.e.,
\begin{eqnarray}
\label{eq-left}
{\sf desc}(\mathcal{C}_0) \cap \mathcal{C} = \{{\bf c}_j \ | \ j \in X\},
\end{eqnarray}
where $X$ is the set in Line $9$ of Algorithm \ref{alg-q-find-inter}.

Based on the above fact, according to Lines $12-17$ of Algorithm \ref{alg-q-find-inter},
we know that $r_{i,k}$, $i \in \{1,2,\ldots,n\}$ and $k \in \{0,1,\ldots,q-1\}$,
is the numbers of occurrences of the elements $k$
in ${\bf c}_{X(1)}(i), {\bf c}_{X(2)}(i), \ldots, {\bf c}_{X(|X|)}(i)$.
According to Lines $18-24$ of Algorithm \ref{alg-q-find-inter},
the set $Y$ in Line $24$ is a subset of $Q=\{0,1, \ldots, q-1\}$ such that
each element  $k \in Y$ occurs exactly once in ${\bf c}_{X(1)}(i), {\bf c}_{X(2)}(i), \ldots, {\bf c}_{X(|X|)}(i)$.
Therefore, according to Lines $25-32$ of Algorithm \ref{alg-q-find-inter},
the output $U$ is a subset of $X$ satisfying  that for any $h\in U$,
there exists an integer $i \in \{1,2,\ldots,n\}$ such that
${\bf c}_h(i) \neq {\bf c}_{h'}(i)$ for any $h' \in X$.
Together with the  fact in Formula \eqref{eq-left},
we know that for any $h\in U$,
there exists an integer $i \in \{1,2,\ldots,n\}$ such that
${\bf c}_h(i) \neq {\bf c}'(i)$ for any ${\bf c}' \in ({\sf desc}(\mathcal{C}_0) \cap \mathcal{C})\setminus \{{\bf c}_h\}$.
Then we can obtain that ${\bf c}_h \in \mathcal{C}_0$.
Otherwise, if ${\bf c}_h \notin \mathcal{C}_0$,
then ${\bf c}_h(i) \notin \mathcal{C}_0(i)$ according to the uniqueness of ${\bf c}_h(i)$.
Thus $({\sf desc}(\mathcal{C}_0)\cap \mathcal{C})(i) \neq \mathcal{C}_0(i)$.
This is a contradiction. So we have showed that $\{{\bf c}_j \ | \ j \in U\} \subseteq \mathcal{C}_0$.

Finally, it suffices to show that  $U$ is not an emptyset.
Since code $\mathcal{C}$ has $t$-uniqueness descendant code,
there exists a codeword ${\bf c}_h \in \mathcal{C}_0$ and $i\in \{1,2,\ldots,n\}$, such that
\begin{equation}\label{eq-unique-word}
{\bf c}_h(i) \neq {\bf c}'(i) \
\textrm{for any} \  {\bf c}' \in ({\sf desc}(\mathcal{C}_0) \cap \mathcal{C}) \setminus \{{\bf c}_h\}.
\end{equation}
This implies that $h \in U$, i.e., $U$ is not an emptyset.
\end{IEEEproof}

\textit{Proof of Theorem \ref{th-soft-TA}: }\
We need to show that for any $\mathcal{C}_0 \subseteq \mathcal{C}$ with $1 \leq |{\mathcal{C}_0}| =t_0 \leq t$
and ${\bf x} = {\sf AT}({\mathcal{C}_0})$ being of the form in Formula (\ref{equa-gene-word}),
the following two conditions are satisfied:
\begin{itemize}
  \item $\phi({\bf x})=\mathcal{C}_0$, where $\mathcal{\phi}$ is Algorithm \ref{alg-soft-TA}.
  \item The computational complexity of Algorithm \ref{alg-soft-TA} is $O(tnM)$.
\end{itemize}

Firstly, consider the computational complexity.
According to Proposition \ref{propo-desc} and Lemma \ref{lem-find-inter},
the computational complexities of  Algorithm \ref{alg-desc} and Algorithm \ref{alg-q-find-inter}
are $O(n)$ and $O(\min\{t,2\}nM)$, respectively.
Together with the fact that $t_0\leq t$,
one can derive that the computational complexity of Algorithm \ref{alg-soft-TA} is $O(tnM)$.

Secondly, we will show that $\phi({\bf x})=\mathcal{C}_0$.
Since $\mathcal{C}$ has $t$-uniqueness descendant code,
according to Proposition \ref{propo-colludersNB},
we know that $t_0 = |{\mathcal{C}_0}|$,
where $t_0$ is the number in Line $1$ of Algorithm \ref{alg-soft-TA}.
The algorithm below will perform iterations.

In the first iteration,
since ${\bf x} = {\sf AT}(\mathcal{C}_0)$,
according to Proposition \ref{propo-desc},
we have ${\bf R} = {\sf desc}(\mathcal{C}_0)$,
where ${\bf R}$ is the set in Line $6$ of Algorithm \ref{alg-soft-TA},
i.e.,  ${\bf R}$ is the output of Algorithm \ref{alg-desc}.
Since $\mathcal{C}$ has $t$-uniqueness descendant code,
according to Lemma \ref{lem-find-inter},
one can obtain that
\begin{eqnarray*}
U' \neq \emptyset \ \mbox{and} \ \mathcal{C}_1=\{{\bf c}_j \ | \ j\in U'\} \subseteq \mathcal{C}_0.
\end{eqnarray*}
where $U'$ is the set in Line $7$ of Algorithm \ref{alg-soft-TA},
i.e.,  $U'$ is the output of Algorithm \ref{alg-q-find-inter}.
In summary, during the first iteration,
one can obtain an index set of  non-empty subset $\mathcal{C}_1$ of $\mathcal{C}_0$
by using the generated word ${\bf x}={\sf AT}(\mathcal{C}_0)$.
\begin{itemize}
  \item If $|U|=t_0$, we can know that $|\mathcal{C}_1|=t_0=|\mathcal{C}_0|$.
Thus $\mathcal{C}_1=\mathcal{C}_0$ as $\mathcal{C}_1\subseteq \mathcal{C}_0$.
That is $\{{\bf c}_j \ | \ j \in U\}=\mathcal{C}_0$.
Therefore,   $\phi({\bf x})=\mathcal{C}_0$.
\item If $|U|< t_0$, then the algorithm proceeds to the second iteration.
\end{itemize}

In the second iteration, one can directly check that
the updated generated word ${\bf x}$ in the left-hand side of the equation  in Line $5$
is the generated word of $\mathcal{C}_0 \setminus \mathcal{C}_1$, i.e.,  ${\bf x} = {\sf AT}(\mathcal{C}_0 \setminus \mathcal{C}_1)$.
Similar to the first iteration, we have
\begin{eqnarray*}
U' \neq \emptyset \ \mbox{and} \ \mathcal{C}_2=\{{\bf c}_j \ | \ j\in U'\} \subseteq
\mathcal{C}_0 \setminus \mathcal{C}_1.
\end{eqnarray*}
where $U'$ is the output of Algorithm \ref{alg-q-find-inter}.
\begin{itemize}
  \item If $|U|=t_0$, we can know that $|\mathcal{C}_2| +|\mathcal{C}_1| = t_0 = |\mathcal{C}_0|$,
  i.e., $|\mathcal{C}_2|=|\mathcal{C}_0|-|\mathcal{C}_1|$.
Thus $\mathcal{C}_2=\mathcal{C}_0 \setminus \mathcal{C}_1$
as $\mathcal{C}_2\subseteq\mathcal{C}_0  \setminus \mathcal{C}_1$.
Thus $\mathcal{C}_1 \cup \mathcal{C}_2 = \mathcal{C}_0$,
i.e., $\{{\bf c}_j \ | \ j \in U\}=\mathcal{C}_0$.
Therefore,  $\phi({\bf x})=\mathcal{C}_0$.
\item If $|U|< t_0$, then the algorithm proceeds to the next iteration.
\end{itemize}

Repeat this process.

In the last iteration,  we can obtain that
the updated generated word ${\bf x}$ in the left-hand side of the equation  in Line $5$
is the generated word of $\mathcal{C}_0 \setminus \cup_{i=1}^{s-1} \mathcal{C}_i$, i.e.,
${\bf x} = {\sf AT}(\mathcal{C}_0 \setminus \cup_{i=1}^{s-1} \mathcal{C}_i)$.
Thus
\begin{eqnarray*}
U' \neq \emptyset \ \mbox{and} \ \mathcal{C}_s=\{{\bf c}_j \ | \ j\in U'\} \subseteq
 \mathcal{C}_0 \setminus \cup_{i=1}^{s-1}\mathcal{C}_i.
\end{eqnarray*}
where $U'$ is the output of Algorithm \ref{alg-q-find-inter}.
Since this is the last iteration, $|U|=t_0$ must hold.
Then $|\mathcal{C}_s| +|\cup_{i=1}^{s-1}\mathcal{C}_i| = t_0 = |\mathcal{C}_0|$,
  i.e., $|\mathcal{C}_s|=|\mathcal{C}_0|-|\cup_{i=1}^{s-1}\mathcal{C}_i|$.
Thus $\mathcal{C}_s=\mathcal{C}_0 \setminus \cup_{i=1}^{s-1}\mathcal{C}_i$
as $\mathcal{C}_2\subseteq\mathcal{C}_0  \setminus \cup_{i=1}^{s-1}\mathcal{C}_i$.
Thus $\cup_{i=1}^{s}\mathcal{C}_i = \mathcal{C}_0$,
i.e., $\{{\bf c}_j \ | \ j \in U\}=\mathcal{C}_0$.
Therefore,  $\phi({\bf x})=\mathcal{C}_0$.

\subsection{Codes with $t$-Uniqueness Descendant Code}

According to Theorem \ref{th-soft-TA},
if an $(n,M,q)$  code has $t$-uniqueness descendant code,
then it can be  used to identify all colluders by using the soft tracing algorithm.
Next, we will find out some codes with $t$-uniqueness descendant code.

\begin{lemma}
\label{lem-SSC-t-unique}
Any $\overline{t}$-SSC$(n,M,q)$ has $t$-uniqueness descendant code.
\end{lemma}
\begin{IEEEproof}
Suppose that $\mathcal{C}$ is a $\overline{t}$-SSC$(n,M,q)$,
and $\mathcal{C}_0\subseteq \mathcal{C}$ with $1 \leq |\mathcal{C}_0|\le t$.
We will show that there exist a codeword ${\bf c} \in \mathcal{C}_0$ and $i \in \{1,2,\ldots,n\}$
such that ${\bf c}(i) \neq {\bf c}'(i)$
for any ${\bf c}' \in ({\sf desc}(\mathcal{C}_0) \cap \mathcal{C}) \setminus \{{\bf c}\}$.
Assume not. Then for any ${\bf c} \in \mathcal{C}_0$ and any $i \in \{1,2,\ldots,n\}$,
there exists ${\bf c}^{(i)} \in ({\sf desc}(\mathcal{C}_0) \cap \mathcal{C}) \setminus \{{\bf c}\}$,
such that ${\bf c}^{(i)}(i)={\bf c}(i)$.
Let $\mathcal{C}_1 = (\mathcal{C}_0\setminus \{{\bf c}\}) \cup (\cup_{i=1}^{n}{\bf c}^{(i)})$.
Then ${\sf desc}(\mathcal{C}_1)={\sf desc}(\mathcal{C}_0)$,
which implies that $\mathcal{C}_1 \in \mathcal{P}(\mathcal{C}_0) = \{ \mathcal{S} \subseteq \mathcal{C} \ | \ {\sf desc}(\mathcal{S}) = {\sf desc}(\mathcal{C}_0) \}$.
Obviously, ${\bf c} \notin \mathcal{C}_1$,
which implies  $\cap_{\mathcal{S} \in \mathcal{P}(\mathcal{C}_0)}\mathcal{S} \neq \mathcal{C}_0$,
a contradiction.
So code $\mathcal{C}$ has $t$-uniqueness descendant code.
\end{IEEEproof}

According to Table \ref{tab-rela-known-codes}, any $t$-FPC$(n,M,q)$ is a $\overline{t}$-SSC$(n,M,q)$.
Thus the following statement always holds.

\begin{corollary}
\label{cor-FPC-t-unique}
Any $t$-FPC$(n,M,q)$ has $t$-uniqueness descendant code.
\end{corollary}

Furthermore, we find that the following code also has $t$-uniqueness descendant code.

\begin{definition}\label{defSMIPPC}
An $(n,M,q)$ code $\mathcal{C}$ is
a \textit{strongly $t$-identifiable parent property code for multimedia fingerprinting},
or $t$-SMIPPC$(n,M,q)$,
if for any subcode $\mathcal{C}_0\subseteq \mathcal{C}$ with $1 \leq |\mathcal{C}_0|\le t$,
we have
$\cap_{\mathcal{S}\in \mathcal{P}(\mathcal{C}_0)}\mathcal{S}\ne \emptyset$,
where
$\mathcal{P}(\mathcal{C}_0)=\{\mathcal{S}\subseteq \mathcal{C} \ | \  {\sf desc}(\mathcal{S})={\sf desc}(\mathcal{C}_0)\}$.
\end{definition}

\begin{lemma}
\label{lem-SMIPPC-unique}
Any $t$-SMIPPC$(n,M,q)$ has $t$-uniqueness descendant code.
\end{lemma}
\begin{IEEEproof}
Suppose that $\mathcal{C}$ is a $t$-SMIPPC$(n,M,q)$,
and $\mathcal{C}_0\subseteq \mathcal{C}$ with $1 \leq |\mathcal{C}_0|\le t$.
We will show that there exist a codeword ${\bf c} \in \mathcal{C}_0$ and $i \in \{1,2,\ldots,n\}$
such that ${\bf c}(i) \neq {\bf c}'(i)$
for any ${\bf c}' \in ({\sf desc}(\mathcal{C}_0) \cap \mathcal{C}) \setminus \{{\bf c}\}$.
Assume not. Then for any ${\bf c} \in \mathcal{C}_0$ and any $i \in \{1,2,\ldots,n\}$,
there exists ${\bf c}^{(i)} \in ({\sf desc}(\mathcal{C}_0) \cap \mathcal{C}) \setminus \{{\bf c}\}$,
such that ${\bf c}^{(i)}(i)={\bf c}(i)$.
Without loss of generality,
we may suppose that $\mathcal{C}_0=\{{\bf c}_1, {\bf c}_2, \ldots, {\bf c}_{t_0}\}$,
where $1 \leq t_0 \leq t$.
Similar to the proof of Lemma \ref{lem-SSC-t-unique},
for any $j \in \{1,2,\ldots, t_0\}$,
there exists $\mathcal{C}_j \in \mathcal{P}(\mathcal{C}_0)
= \{ \mathcal{S} \subseteq \mathcal{C} \ | \ {\sf desc}(\mathcal{S}) = {\sf desc}(\mathcal{C}_0) \}$
and ${\bf c}_j \notin \mathcal{C}_j$.
Then ${\bf c}_j \notin \cap_{\mathcal{S} \in \mathcal{P}(\mathcal{C}_0)}\mathcal{S}$.
Due to the arbitrariness of $j$,
we conclude that ${\bf c}\notin \cap_{\mathcal{S} \in \mathcal{P}(\mathcal{C}_0)}\mathcal{S}$
for any ${\bf c} \in \mathcal{C}_0$.
On the other hand, it is obvious that $\mathcal{C}_0 \in \mathcal{P}(\mathcal{C}_0)$,
which implies that $\cap_{\mathcal{S} \in \mathcal{P}(\mathcal{C}_0)}\mathcal{S} \subseteq \mathcal{C}_0$.
Therefore, $\cap_{\mathcal{S} \in \mathcal{P}(\mathcal{C}_0)}\mathcal{S} = \emptyset$,
a contradiction to that $\mathcal{C}$ is a $t$-SMIPPC$(n,M,q)$.
So code $\mathcal{C}$ has $t$-uniqueness descendant code.
\end{IEEEproof}

\begin{remark}
According to Theorem \ref{th-soft-TA},
any $t$-SMIPPC$(n,M,2)$ can be  used to identify all colluders by using the soft tracing algorithm.
However, it is shown that at least one colluder can be identified with the tracing algorithm in \cite{JGC}.
This provides more evidence on the powerful function of the soft tracing algorithm resisting the averaging attacks.
\end{remark}

Now we can extend  Table \ref{tab-rela-known-codes} to
Table \ref{tab-rela-all-codes} by adding the concept of an SMIPPC,
where the relationship between an SSC and an SMIPPC can be found in Lemma \ref{lem-SSC-SMIPPC}.

\begin{lemma} \rm{(\cite{JGC})}
\label{lem-SSC-SMIPPC}
Any $\overline{t}$-SSC$(n,M,q)$ is a $t$-SMIPPC$(n,M,q)$.
\end{lemma}

{\begin{table}[h]
\begin{center}
\caption{ Relationships among different types of fingerprinting codes}
\label{tab-rela-all-codes}
   \begin{tabular}{ccccccccc}
     &  & $t$-FPC$(n,M,q)$ &  $ \Longrightarrow $ & $\overline{t}$-SSC$(n,M,q)$ & $\Longrightarrow$ & $t$-SMIPPC$(n,M,q)$ \\
     & & $ \Downarrow $  &  & $ \Downarrow $ &  & $\Downarrow$ \begin{footnotesize} $q=2$ \end{footnotesize}\\
     & & SCLD$(n,M,q; L)$  & $ \Longrightarrow $ & $\overline{t}$-SC$(n,M,q)$ &
     $ \overset{q=2}{\Longrightarrow}  $ & AACC$(n,M,2)$
  \end{tabular}
\end{center}
\end{table}}

We also extend  Table \ref{tab-comp-konwn-codes} to
Table \ref{tab-comple-all-codes} by adding the concept of an SMIPPC.
{\begin{table*}[h]
\center
\caption{Comparison of computational complexities of tracing algorithms of binary fingerprinting codes}
\label{tab-comple-all-codes}
  \begin{tabular}{|c|c|c|c|c|c|c|c|c|}
\hline
 & FPC &   SSC  &   SC &  SCLD  & SMIPPC \\ \hline
\tabincell{c}{Complexity}
    &   $O(nM)$
    &   $O(nM)$
    &     $O(nM^t)$
    &   $O(\max\{nM,nL^t\})$
    &  $O(tnM)$  \\  \hline

Reference & \cite{Bla,CM} &  \cite{JCM}  &   \cite{CJM,CM}
&    \cite{GVM}& \tabincell{c}{Lemma \ref{lem-SMIPPC-unique},\\Theorem \ref{th-soft-TA}}\\ \hline
\end{tabular}
\end{table*}

In practice, the maximum number $t$ of colluders is very small comparing  with the size $M$ of the code.
Thus a $t$-SMIPPC$(n, M, 2)$ has the same traceability as a $t$-FPC$(n, M, 2)$
(or a $\overline{t}$-SSC$(n, M, 2)$), and better traceability than
those of a $\overline{t}$-SC$(n, M, 2)$ and a $\overline{t}$-SCLD$(n, M, 2; L)$ with $L > M^{\frac{1}{t}}$.

Finally,  we list the state-of-the-art code rates about FPCs, SSCs, SCs, SCLDs and SMIPPCs
in Table \ref{tab-comp-rates},
where $R_{\rm SMIPPC}(t,n) = \limsup_{q \rightarrow \infty}\frac{\log_q M_{\rm SMIPPC}(t, n, q)}{n}$, and
$M_{\rm SMIPPC}(t, n, q)$  denote the largest cardinality of a $q$-ary $t$-SMIPPC of length $n$.
According to this table, we know that a $t$-SMIPPC$(n,M,q)$ has the best code rate among these fingerprinting codes.}

{\begin{table*}[h]
\center
\caption{Comparison of code rates among different types of $q$-ary codes when $q\rightarrow \infty$}
\label{tab-comp-rates}
  \begin{tabular}{|c|c|c|c|c|c|c|c|}
\hline
 & $R_{\rm FPC}(t,n)$
 &  $R_{\rm SSC}(\overline{t},n)$, $R_{\rm SC}(\overline{t},n)$, $R_{\rm SCLD}(\overline{t},n;L)$
 &   $R_{\rm SMIPPC}(t,n)$ \\ \hline

\tabincell{c}{Code Rate}
    &   $=\frac{\lceil n/t \rceil}{n}$
    &  \tabincell{c}
    {$\leq
      \left\{\begin{array}{rl}
         \frac{\lceil 2n/3 \rceil}{n}, \  &   \mbox{if}  \ t=2,\\[2pt]
         \frac{\lceil n/(t-1) \rceil}{n},  &  \mbox{if}  \ t>2.\\[2pt]
             \end{array}
      \right.$}

    &  $\geq\frac{t}{2t-1}$   \\  \hline

Reference & \cite{Bla} & \cite{Bl2,GVM,JCM}  &  \cite{JGC} \\ \hline
\end{tabular}
\end{table*}

\section{Two-Stage Soft Tracing Algorithm}   %
\label{sec-two-STA}

As stated in \cite{For}, concatenation construction is a
powerful method to construct infinite families of codes with a required property by combining
a seed code with the property over a small alphabet, together with an appropriate code over
a large alphabet whose size is the size of the seed code.

Suppose that ${\mathcal B}=\{{\bf b}_1,\ldots, {\bf b}_M\}$ is an  $(n_1,M,q)$ code,
and ${\mathcal D}=\{{\bf d}_1, {\bf d}_2, \ldots, {\bf d}_q\}$ is an $(n_2,q,2)$ code.
Then we construct an $(n_1n_2, M, 2)$ code by concatenating ${\mathcal B}$ with ${\mathcal D}$ as follows.
Let $f: \{0,1,\ldots,q-1\} \rightarrow {\mathcal D}$ be a bijective mapping such that $f(k)={\bf d}_{k+1}$.
For any codeword ${\bf b} = ({\bf b}(1),{\bf b}(2),\ldots,{\bf b}(n_1))^{T} \in {\mathcal B}$,
we define $f({\bf b})=(f({\bf b}(1)),f({\bf b}(2)),\ldots, f({\bf b}(n_1)))^{T} $.
Obviously, $f({\bf b})$ is a binary vector of length $n_1n_2$.
We define a new $(n_1n_2,M,2)$ code
\begin{equation}
\label{equ-concat-code}
 \mathcal{C}=\{f({\bf b}_1), f({\bf b}_2), \ldots, f({\bf b}_M)\},
\end{equation}
denoted by $\mathcal{C} = \mathcal{B}\circ \mathcal{D}$.

\begin{example} Let
\begin{eqnarray*}
\mathcal{B}&=& \left(
  \begin{array}{cccccccc}
  {\bf b}_1 &{\bf b}_2 & {\bf b}_3 & {\bf b}_4 & {\bf b}_5 &{\bf b}_6 \\ \hline
     0 & 0 & 1 & 1 & 2 & 2 \\
    0 & 1 & 1 & 2 & 2 & 0\\
  \end{array}
\right), \ \ \ \ \ \ \
{\mathcal D}=
\left(
  \begin{array}{ccc}
  {\bf d}_1 &{\bf d}_2 & {\bf d}_3 \\ \hline
    0 & 1 & 0 \\
    0 & 0 & 1 \\
  \end{array}
\right).
\end{eqnarray*}

Concatenate ${\mathcal B}$ with ${\mathcal D}$, we obtain

\begin{eqnarray*}
\mathcal{C} = \mathcal{B}\circ \mathcal{D}=
\left(
  \begin{array}{cccccc}
  {\bf c}_1 &{\bf c}_2 & {\bf c}_3 & {\bf c}_4 & {\bf c}_5 &{\bf c}_6 \\ \hline
    0 & 0 & 1 & 1 & 0 & 0 \\
    0 & 0 & 0 & 0 & 1 & 1\\
    0 & 1 & 1 & 0 & 0 & 0\\
    0 & 0 & 0 & 1 & 1 & 0\\
  \end{array}.
\right)
\end{eqnarray*}
\end{example}

Clearly, each codeword of the code $\mathcal{D}$ may be used several times
to obtain the concatenated code $\mathcal{C}$.
Hence, we have to consider multi-set which is also useful to prove Theorem \ref{th-two-sta-STA}.
In order to distinguish simple set and multi-set, we use a square bracket to denote a multi-set.
For example, a  multi-set $\mathcal{D}_0=\{{\bf d}_1,{\bf d}_1,{\bf d}_2,{\bf d}_2,{\bf d}_2,{\bf d}_3\}$
will be written as $[\mathcal{D}_0] = [ 2\times {\bf d}_1, 3\times {\bf d}_2, 1\times {\bf d}_3 ]$.
For the multi-set  $[\mathcal{D}_0] = [r_{1}\times {\bf d}_1, r_{2}\times {\bf d}_2, \ldots, r_{s}\times {\bf d}_s]$,
the \textit{size} of $[\mathcal{D}_0]$ is  denoted by $|[\mathcal{D}_0]| = \sum_{j=1}^{s}r_j$.
Furthermore, if each element of $[\mathcal{D}_0]$ is contained in  $\mathcal{D}$,
we can use the notation $[\mathcal{D}_0] \subseteq \mathcal{D}$.
Similar to the case of simple set,
the generated code ${\bf x}$ of $[\mathcal{D}_0]$ is
\begin{equation}\label{eq-multi-generated-word}
{\bf x} = {\sf AT}([\mathcal{D}_0])= \frac{1}{r_1 + r_2 + \ldots + r_s}\sum\limits_{j=1}^{s}r_j{\bf d}_{j}.
\end{equation}
For convenience,
let ${\bf x}$ be of the form in Formula \eqref{equa-gene-word}.

Now, we will introduce a two-stage soft tracing algorithm (Algorithm \ref{alg-two-stage-STA})
for concatenated codes.
Here, we only illustrate the ideas of our algorithms,
and we will establish the conditions for the algorithms' validity later.

\begin{itemize}
\item[1)] Algorithm \ref{alg-multi-soft-TA}:
Suppose that ${\mathcal D}$ is an $(n,M,2)$ code,
$[\mathcal{D}_0] = [r_{1}\times {\bf d}_1, r_{2}\times {\bf d}_2, \ldots, r_{s}\times {\bf d}_s]
\subseteq \mathcal D$, and ${\bf x} = {\sf AT}([\mathcal{D}_0])$
being of the form in Formula (\ref{equa-gene-word}).
When the inputs are ${\bf x}$ and the size of $[\mathcal{D}_0]$,
we expect the algorithm to output the index set of the multi-set $[\mathcal{D}_0]$.
During the first iteration,
we know that the output of Algorithm \ref{alg-desc}
is ${\sf desc}(\{{\bf d}_1, {\bf d}_2, \ldots, {\bf d}_s\})$,
and the output of Algorithm \ref{alg-q-find-inter}
is an index set of a subset $[\mathcal{D}_1]$ of $[\mathcal{D}_0]$.
For the next iterations, we can regard $\mathcal{D}_1$ as a multi-set $[\mathcal{D}_1]$.
In the second iteration, the generated word is updated, i.e.,
${\bf x}$ is in fact the generated words of $[\mathcal{D}_0] \setminus [\mathcal{D}_1]$.
Similarly, Algorithm \ref{alg-q-find-inter} outputs an index set of  a subset
$[\mathcal{D}_2] \subseteq [\mathcal{D}_0] \setminus [\mathcal{D}_1]$.
Repeat this process.
In the last iteration, Algorithm \ref{alg-q-find-inter} outputs an index of a subset
$[\mathcal{D}_s] \subseteq [\mathcal{D}_0] \setminus \cup_{i=1}^{s-1}[\mathcal{D}_i]$.
Now we expect that $\mathcal[{D}_0] = \cup_{i=1}^{s}[\mathcal{D}_i]$.
In a word, we expect $[{\bf d}_j \ | \ j\in[U]]$ is equal to $[\mathcal{D}_0]$,
where $[U]$ is the output of the algorithm.

\item[2)] Algorithm \ref{alg-two-stage-STA}:
Suppose $\mathcal{C} = \mathcal{B}\circ \mathcal{D}$
is an $(n_1n_2, M, 2)$  code  in Formula \eqref{equ-concat-code}.
Let $\mathcal{C}_0\subseteq \mathcal{C}$
and ${\bf x}={\sf AT}(\mathcal{C}_0)$ being of the form in Formula (\ref{equa-gene-word}).
According to the construction in Formula \eqref{equ-concat-code},
we know that there exists a subset $\mathcal{B}_0 \subseteq \mathcal{B}$
with $|\mathcal{B}_0|=|\mathcal{C}_0|$,
such that $\mathcal{C}_0 = \mathcal{B}_0\circ \mathcal{D}$.
When the inputs are ${\bf x}$ and $n_1$,
we expect the algorithm to output the index set of $\mathcal{C}_0$.
We first determine the exact number of colluders,
i.e., we expect $t_0 =$ lcm$(t_1, t_2, \ldots, t_{n})$ is equal to the size of $\mathcal{C}_0$.
Next, the algorithm will enter the while loop.
For convenience, let $[\mathcal{D}_0^{(i)}]=[f({\bf b}(i)) \ | \ {\bf b} \in \mathcal{B}_0]$
for any  $i \in \{1,2,\ldots,n_1\}$.
During the first iteration of the while loop,
after the for loop,
we expect to  obtain $\mathcal{B}_0(i)$ for any $i \in \{1,2,\ldots,n_1\}$,
i.e., we expect that $\mathcal{B}_0(i)$ is equal to ${\bf R}(i)$,
where ${\bf R}(i)$ is the set in Line $10$ of the algorithm.
Hence  ${\sf desc(\mathcal{B}_0}) ={\bf R}$.
With the input ${\sf desc(\mathcal{B}_0}) ={\bf R}$,
Algorithm \ref{alg-q-find-inter} outputs an index set of a subset $\mathcal{B}_1$ of $\mathcal{B}_0$.
According to the construction in Formula \eqref{equ-concat-code},
it is obvious that such an index set is also an index of the subset $\mathcal{C}_1$ of $\mathcal{C}_0$,
where $\mathcal{C}_1= \mathcal{B}_1 \circ \mathcal{D}$.
In the second iteration of the while loop, the generated word is updated, i.e.,
${\bf x}$ is in fact the generated words of $\mathcal{C}_0 \setminus \mathcal{C}_1$.
Similarly, Algorithm \ref{alg-q-find-inter} outputs an index set of a subset
$\mathcal{C}_2 \subseteq \mathcal{C}_0 \setminus \mathcal{C}_1$.
Repeat this process.
In the last iteration, Algorithm \ref{alg-q-find-inter} outputs an index set of a subset
$\mathcal{C}_s \subseteq \mathcal{C}_0 \setminus \cup_{i=1}^{s-1}\mathcal{C}_i$.
Now we expect that $\mathcal{C}_0 = \cup_{i=1}^{s}\mathcal{C}_i$.
In a word, we expect $\{{\bf c}_j \ | \ j \in U\}$ is equal to $\mathcal{C}_0$,
where $U$ is the output of the algorithm.
\end{itemize}

\begin{algorithm}[H]
\caption{ {\tt Multi-set Soft Tracing Algorithm}}
\label{alg-multi-soft-TA}
\SetKwData{Left}{left}\SetKwData{This}{this}\SetKwData{Up}{up}
\SetKwFunction{Union}{Union}\SetKwFunction{FindCompress}{FindCompress}
\SetKwInOut{Input}{input}\SetKwInOut{Output}{output}

\Input {${\bf x} = ({\bf x}(1), {\bf x}(2), \ldots, {\bf x}(n))$, $|[{\mathcal D}_0]|$}

Denote $t_0 = |[{\mathcal D}_0]|$\;

$[U]=\emptyset$;

$Flag=True$;

\While {$Flag$ and $|[U]| < t_0$}
    {   ${\bf x}=\frac{t_0{\bf x}-\sum_{j\in [U]}{\bf c}_{j}}{t_0-|[U]|}$\;

        Execute Algorithm \ref{alg-desc}
        with the input ${\bf x}$,
        and the output is ${\bf R}$\;

        Execute Algorithm \ref{alg-q-find-inter} with the input ${\bf R}$, and the output is $U'$\;

        \If {$U'=\emptyset$}
            {   $Flag = False$ \;
            }

        \Else
            {$[U]=[U]\cup [U']$\;}

         }

 \Output {the index set $[U]$}

\end{algorithm}
\vskip 1cm

\begin{algorithm}[H]
\caption{ {\tt Two-Stage Soft Tracing Algorithm}}
\label{alg-two-stage-STA}
\SetKwData{Left}{left}\SetKwData{This}{this}\SetKwData{Up}{up}
\SetKwFunction{Union}{Union}\SetKwFunction{FindCompress}{FindCompress}
\SetKwInOut{Input}{input}\SetKwInOut{Output}{output}

\Input {${\bf x}= ({\bf x}(1), {\bf x}(2), \ldots, {\bf x}(n)), n_1$}

Denote $t_0 =$ lcm$(t_1, t_2, \ldots, t_{n})$, $n_2 = \frac{n}{n_1}$\;

$Flag=True$\;

$[U]=\emptyset$\;

\While{Flag and $U< t_0$}
    {   ${\bf x}=\frac{t_0{\bf x}-\sum_{j\in U}{\bf c}_{j}}{t_0-|U|}$\;

        \For { $i=1$  {\bf to }  $n_1$ }
            {
                ${\bf x}^{(i)} = ({\bf x}((i-1)n_2+1), {\bf x}((i-1)n_2+2), \ldots, {\bf x}(in_2))$\;
                Execute Algorithm \ref{alg-multi-soft-TA} with the input $({\bf x}^{(i)}, t_0)$,
                and the output is $[U']$\;
                \If {$|[U']| = t_0$}
                    {   ${\bf R}(i)=\{j-1 \ | \ j\in [U']\}$\;}

                \Else
                    {$Flag=False$\;}
            }

            \If{$Flag=True$}
                {   Execute Algorithm \ref{alg-q-find-inter} with the input
                    ${\bf R} = {\bf R}(1) \times {\bf R}(2) \times \ldots \times {\bf R}(n_1)$,
                    and the output is $U''$\;

                    \If {$U''=\emptyset$}
                    {   $Flag = False$ \;}

                    \Else
                        { $U=U\cup U''$}
                }
    }

\If {$|U| \neq t_0$}
    {
       \Output {       This code does not satisfy the conditions of the algorithm.  }
    }
\Else
    { \Output {the index set $U$}}

  \end{algorithm}
\vskip 1cm

We now characterize the codes that satisfy the algorithm's validity criteria.

\begin{theorem}
\label{th-two-sta-STA}
Suppose that  $\mathcal{C} = \mathcal{B}\circ \mathcal{D}$ is of the form in Formula \eqref{equ-concat-code},
and both of the  codes ${\mathcal B}$ and ${\mathcal D}$ have $t$-uniqueness descendant codes.
Under the assumption that the number of colluders in the averaging attack is at most $t$,
the code $\mathcal{C}$ can be applied to identify all colluders by using the two-stage soft tracing algorithm (Algorithm \ref{alg-two-stage-STA}),
and the computational complexity is $O(t^2n_1n_2q + \min\{t,q\}tn_1M)$.
\end{theorem}

Similar to the previous section, in order to prove the above theorem, we would like to determine the exact number of the colluders first.

\begin{lemma}
\label{find-multi-SMIPPC-binary}
Let $\mathcal{D}$ be a $(n, M, 2)$ code with $t$-uniqueness descendant code.
Suppose that $[\mathcal{D}_0] = [r_{1}\times {\bf d}_1, r_{2}\times {\bf d}_2, \ldots, r_{s}\times {\bf d}_s] \subseteq \mathcal{D}$ with $1 \leq |[\mathcal{D}_0]| \leq t$,
${\bf x}= {\sf AT}([\mathcal{D}_0])$ being of the form in Formula (\ref{equa-gene-word}),
and $L ={\sf lcm}(t_1, t_2, \ldots, t_n)$, i.e., the least common multiple.
Then there exists a positive integer $b$, such that $| [\mathcal{D}_0] | = bL$.
Furthermore, $b \mid r_i$ holds for any $i \in \{1,2,\ldots,s\}$.
\end{lemma}
\begin{IEEEproof} Firstly, according to Formulas \eqref{equa-gene-word} and \eqref{eq-multi-generated-word},
$| [\mathcal{D}_0] |$ should be a multiple of $t_i$ for any  $i \in \{1,2,\ldots,n\}$.
Together with the fact that $L ={\sf lcm}(t_1, t_2, \ldots, t_n)$,
one can know  that there exists a positive integer $b$, such that $| [\mathcal{D}_0] | = bL$.
Hence
\begin{eqnarray}
\label{num-multi-set}
\sum_{j=1}^{s}r_j = | [\mathcal{D}_0] | = bL.
\end{eqnarray}

Next ${\bf x}$ can be written as
\begin{equation}\label{equa-gene-word-1}
{\bf x} = (\frac{a_1}{t_1}, \frac{a_2}{t_2}, \ldots, \frac{a_n}{t_n})=
(\frac{b_1}{L}, \frac{b_2}{L}, \ldots,  \frac{b_{n}}{L}).
\end{equation}
where $b_i = \frac{La_i}{t_i}$ for any $i \in \{1,2,\ldots,n\}$.
Furthermore,
${\bf x}$ can also be written as
\begin{equation}\label{equa-gene-word-1}
{\bf x} = (\frac{bb_1}{bL}, \frac{bb_2}{bL}, \ldots,  \frac{bb_{n}}{bL}).
\end{equation}
Thus one can know that
\begin{eqnarray}
\label{eq-conca-coord}
\sum_{j=1}^{s}r_j{\bf d}_j(i)=bb_i
\end{eqnarray}
holds for any  $i \in \{1,2,\ldots,n\}$ as the number of colluders $| [\mathcal{D}_0] |$ is $bL$.

Let $\mathcal{D}_1 = \{ {\bf d}_1, {\bf d}_2, \ldots, {\bf d}_s\}$.
Since $\mathcal{D}$ has $t$-uniqueness descendant code and $s\leq t$, there exist  ${\bf d} \in \mathcal{D}_1$
and $i \in \{1,2,\ldots,n\}$ such that
${\bf d}(i) \neq {\bf d}'(i)$ for any ${\bf d}' \in \mathcal{D}_1 \setminus \{{\bf d}\}$.
Without loss of generality, suppose that ${\bf d} = {\bf d}_1$ and $i = 1$,
i.e.,  ${\bf d}_1(1) \neq {\bf d}'(1)$ for any ${\bf d}' \in \mathcal{D}_1 \setminus \{{\bf d_1}\}$.
Similarly, we can suppose that ${\bf d}_i(i) \neq {\bf d}'(i)$
for any ${\bf d}' \in \mathcal{D}_{i} \setminus \{{\bf d}_{i}\}$,
where $\mathcal{D}_{i} = \{ {\bf d}_{i}, {\bf d}_{i+1}, \ldots, {\bf d}_{s}\}$ and
$i \in \{2,3,\ldots,s-1\}$.
\begin{itemize}
  \item Consider the first row of $[\mathcal{D}_0]$.
  \begin{itemize}
  \item[1)] If ${\bf d}_1(1) = 1$,  then ${\bf d}_j(1) = 0$ for any $j \in \{2,3,\ldots,s\}$.
   According to Formula \eqref{eq-conca-coord},
  one can obtain that $r_1{\bf d}_1(1)=bb_1$, i.e., $r_1=bb_1$, which implies that $b \mid r_1$.
   \item[2)] If ${\bf d}_1(1) = 0$, then ${\bf d}_j(1) = 1$ for $j \in \{2,3,\ldots,s\}$.
   According to Formula \eqref{eq-conca-coord},
   one can obtain $\sum_{j=2}^{s}r_j{\bf d}_j(1)=bb_1$, i.e., $\sum_{j=2}^{s}r_j=bb_1$.
   On the other hand, by using the fact $\sum_{j=1}^{s}r_j = bL$ in Formula \eqref{num-multi-set},
   we have $r_1=bL-\sum_{j=2}^{s}r_j=bL-bb_1=b(L-b_1)$, which implies $b \mid r_1$.
  \end{itemize}
\item Consider the second row of $[\mathcal{D}_0]$.
  \begin{itemize}
  \item[1)] If ${\bf d}_2(2) = 1$,  then ${\bf d}_j(2) = 0$ for any $j \in \{3,4,\ldots,s\}$.
   According to Formula \eqref{eq-conca-coord},
  one can obtain that $r_1{\bf d}_1(2) + r_2{\bf d}_2(2)=bb_2$,
  i.e., $r_1{\bf d}_1(2) + r_2=bb_i$.
  Hence $r_2 =bb_i-r_1{\bf d}_1(2)$. Together with the fact $b \mid r_1$,
   one can obtain that $b \mid r_2$.
  \item[2)] If ${\bf d}_2(2) = 0$, then ${\bf d}_j(2) = 1$ for any $j \in \{3,4,\ldots,s\}$.
   According to Formula \eqref{eq-conca-coord},
   one can obtain $r_1{\bf d}_1(2)+\sum_{j=3}^{s}r_j{\bf d}_j(2)=bb_2$,
   i.e., $r_1{\bf d}_1(2)+\sum_{j=3}^{s}r_j=bb_2$.
   On the other hand, by using  the fact $\sum_{j=1}^{s}r_j = bL$ in Formula \eqref{num-multi-set},
   we have $r_2=bL-r_1-\sum_{j=3}^{s}r_j=bL-r_1-bb_2 + r_1{\bf d}_1(2)$,
   which implies $b \mid r_2$ as $b \mid r_1$.
  \end{itemize}
 $\vdots$
 \item Consider the $(s-1)$th row of $[\mathcal{D}_0]$.
 \begin{itemize}
  \item[1)] If ${\bf d}_{s-1}(s-1) = 1$, then ${\bf d}_s(s-1) = 0$.
   According to Formula \eqref{eq-conca-coord},
  one can obtain that $\sum_{j=1}^{s-1}r_j{\bf d}_j(s-1)=bb_{s-1}$,
  i.e., $\sum_{j=1}^{s-2}r_j{\bf d}_j(s-1) + r_{s-1}=bb_{s-1}$.
  Hence $r_{s-1} =bb_{s-1}-\sum_{j=1}^{s-2}r_j{\bf d}_j(s-1)$.
  Together with the facts $b \mid r_j$ for any $j \in \{1,2,\ldots,s-2\}$,
   one can obtain that $b \mid r_{s-1}$.
  \item[2)] If ${\bf d}_{s-1}(s-1) = 0$, then ${\bf d}_s(s-1) = 1$.
  According to Formula \eqref{eq-conca-coord},
   one can obtain $\sum_{j=1}^{s-2}r_j{\bf d}_j(s-1) + r_s{\bf d}_s(s-1)=bb_{s-1}$,
   i.e., $\sum_{j=1}^{s-2}r_j{\bf d}_j(s-1) + r_s=bb_{s-1}$.
   On the other hand, by using  the fact $\sum_{j=1}^{s}r_j = bL$ in Formula \eqref{num-multi-set},
   we have $r_{s-1} = bL - \sum_{j=1}^{s-2}r_j - r_s
   = bL - \sum_{j=1}^{s-2}r_j - bb_{s-1} + \sum_{j=1}^{s-2}r_j{\bf d}_j(s-1)$.
   Together with the facts $b \mid r_j$ for any  $j \in \{1,2,\ldots,s-2\}$,
   one can obtain that $b \mid r_{s-1}$.
  \end{itemize}
\item Consider the $s$th row of $[\mathcal{D}_0]$.
According to the fact $\sum_{j=1}^{s}r_j = bL$ in Formula \eqref{num-multi-set},
   we have $r_s=bL-\sum_{j=1}^{s-1}r_j$.
   Together with the facts $b \mid r_j$ for any  $j \in \{1,2,\ldots,s-1\}$,
   one can obtain that $b \mid r_s$.

\end{itemize}
According to the above discussions,  the conclusion is true.
\end{IEEEproof}

Now, we can determine the exact number of the colluders.

\begin{proposition}
\label{propo-multi-No-coll}
Let  $\mathcal{C} = \mathcal{B}\circ \mathcal{D}$ be of the form in Formula \eqref{equ-concat-code},
and both of the  codes ${\mathcal B}$ and ${\mathcal D}$ have $t$-uniqueness descendant codes.
Suppose that  $\mathcal{C}_0 \subseteq \mathcal{C}$ with $1 \leq | \mathcal{C}_0 | \leq t$, and
${\bf x} = {\sf AT}(\mathcal{C}_0)$ being of the form in Formula (\ref{equa-gene-word}) with $n = n_1n_2$,
Then $|\mathcal{C}_0| =L$, where $L = {\sf lcm}(t_1, t_2, \ldots, t_{n_1n_2})$.
\end{proposition}
\begin{IEEEproof}
Since ${\bf x} = {\sf AT}(\mathcal{C}_0)$ being of the form in Formula (\ref{equa-gene-word})
 and $L ={\sf lcm}(t_1, t_2, \ldots, t_{n_1n_2})$,
$|\mathcal{C}_0|$ should be a multiple of $L$.
Assume that $|\mathcal{C}_0| \neq L$.
Then there exists a positive integer $b$ with $b \geq 2$ such that $|\mathcal{C}_0| = bL$.

Since $\mathcal{C} = \mathcal{B}\circ \mathcal{D}$,
there exists a subset $\mathcal{B}_0 \subseteq \mathcal{B}$,
such that $\mathcal{C}_0 = \mathcal{B}_0 \circ \mathcal{D}$.
For convenience, suppose that
\begin{equation}\label{eq-outcode}
  [\mathcal{B}_0(i)] = [r_{i,0}\times 0, r_{i,1}\times 1, \ldots, r_{i,q-1}\times (q-1)],
\end{equation}
and
\begin{equation}\label{eq-innerword}
  [\mathcal{D}_0^{(i)}] = [f(j) \ | \ j \in [\mathcal{B}_0(i)]]
= [r_{i,0}\times {\bf d}_1, r_{i,1}\times {\bf d}_2, \ldots, r_{i,q-1}\times {\bf d}_q],
\end{equation}
where $r_{i,k}=0$ when $k \notin [\mathcal{B}_0(i)]$,
$i \in \{1,2,\ldots,n_1\}$ and $k \in \{0,1,\ldots,q-1\}$.

For any $i \in \{1,2,\ldots,n_1\}$,
let $L_i = {\sf lcm}(t_{(i-1)n_2+1}, t_{(i-1)n_2+2},\ldots, t_{in_2})$.
Then $L_i \leq L$.
In addition, according to  Lemma \ref{find-multi-SMIPPC-binary},
 there exists a positive integer $b_i$, such that $| [\mathcal{D}_0^{(i)}] | = b_iL_i$.
Furthermore, $b_i \mid r_{i,k}$ holds for any $k \in \{0,1,\ldots,q-1\}$.
According to the construction $\mathcal{C}_0 = \mathcal{B}_0 \circ \mathcal{D}$,
we know that $|\mathcal{C}_0| = | [\mathcal{D}_0^{(i)}] |$, i.e., $bL=b_iL_i$.
Together with the above facts $L_i \leq L$ and $b \geq 2$,
we have $b_i \geq b \geq 2$.
Again, according to the above facts $b_i \mid r_{i,k}$ for any $k \in \{0,1,\ldots,q-1\}$,
we have that $r_{i,k} \neq 0$ implies $r_{i,k} \geq 2$ for any $k \in \{0,1,\ldots,q-1\}$.
Hence  there exists no codeword ${\bf b} \in \mathcal{B}_0$ and $i \in \{1,2,\ldots,n_1\}$
such that ${\bf b}(i) \neq {\bf b}'(i)$ for any ${\bf b}' \in \mathcal{B}_0 \setminus \{{\bf b}\}$
according to Formula \eqref{eq-outcode}.
This contradicts the hypothesis that the code ${\mathcal B}$ has $t$-uniqueness descendant code.
So $|\mathcal{C}_0| = L$.
\end{IEEEproof}

\begin{proposition}
\label{lem-find-multi-STA}
Let $\mathcal{D}$ be an $(n, M, 2)$ code with $t$-uniqueness descendant code.
Suppose that $[\mathcal{D}_0] \subseteq \mathcal{D}$
with $1 \leq |[\mathcal{D}_0]| \leq t$ and ${\bf x} = {\sf AT}([\mathcal{D}_0])$
being of the form in Formula (\ref{equa-gene-word}).
If the inputs are  ${\bf x}$ and $|[\mathcal{D}_0]|$,
then Algorithm \ref{alg-multi-soft-TA} outputs the index set of $[\mathcal{D}_0]$,
i.e., $[{\bf d}_j \ | \ j\in[U]]=[\mathcal{D}_0]$.\,
where $[U]$ is the output of Algorithm  \ref{alg-multi-soft-TA}.
In addition, the computational complexity is $O(tnM)$.
\end{proposition}
\begin{IEEEproof}
Firstly, consider the computational complexity.
According to Proposition \ref{propo-desc} and Lemma \ref{lem-find-inter},
the computational complexities of Algorithm \ref{alg-desc} and Algorithm \ref{alg-q-find-inter}
are $O(n)$ and $O(\min\{t,2\}nM)$, respectively.
In addition, since $1 \leq |[\mathcal{D}_0]| \leq t$ according to the hypothesis of the lemma,
we know that $t_0 \leq t$ where $t_0$ is the parameter in Line $1$ of Algorithm \ref{alg-multi-soft-TA}.
Thus the computational complexity is $O(tnM)$.

Similar to the case of simple subcode,
the set ${\bf R}$ in Line $6$ of Algorithm \ref{alg-multi-soft-TA}
is in fact ${\sf desc}([\mathcal{D}_0])$,
and $[{\bf d}_j \ | \ j\in U'] \subseteq [\mathcal{D}_0]$,
where $U'$ is the set in Line $7$ of Algorithm \ref{alg-multi-soft-TA}, i.e.,  the output of Algorithm \ref{alg-q-find-inter}.
So one can continue the loop until all the codewords in  $[\mathcal{D}_0]$ are obtained.
\end{IEEEproof}

\textit{Proof of Theorem \ref{th-two-sta-STA}: }
We need to show that for any $\mathcal{C}_0 \subseteq \mathcal{C}$ with $1 \leq |{\mathcal{C}_0}| =t_0 \leq t$
and ${\bf x} = {\sf AT}({\mathcal{C}_0})$ being of the form in Formula (\ref{equa-gene-word}),
the following two conditions are satisfied:
\begin{itemize}
  \item $\phi({\bf x})=\mathcal{C}_0$, where $\mathcal{\phi}$ is Algorithm \ref{alg-two-stage-STA}.
  \item The computational complexity of Algorithm \ref{alg-two-stage-STA}
  is $O(t^2n_1n_2q + \min\{t,q\}tn_1M)$.
\end{itemize}

Firstly, consider the computational complexity.
Since $\mathcal{D}$ is an $(n_2,q,2)$ code with $t$-uniqueness descendant code,
according to  Proposition \ref{lem-find-multi-STA},
the computational complexity of Algorithm \ref{alg-multi-soft-TA} is $O(tn_2q)$.
Hence the computational complexity  of Lines $6-15$
in Algorithm \ref{alg-two-stage-STA} is $O(tn_1n_2q)$.
Together with the fact that the computational complexity of Algorithm \ref{alg-q-find-inter} is
$O(\min\{t,q\}n_1M)$,
we have  the computational complexity of Algorithm \ref{alg-two-stage-STA}
is $O(t^2n_1n_2q + \min\{t,q\}tn_1M)$.

Secondly, we will show that $\phi({\bf x})=\mathcal{C}_0$.
Since $\mathcal{C}_0 \subseteq \mathcal{C}$,
there exists a subset $\mathcal{B}_0 \subseteq \mathcal{B}$,
such that $\mathcal{C}_0=\mathcal{B}_0\circ \mathcal{D}$.
Since both of the  codes ${\mathcal B}$ and ${\mathcal D}$ have $t$-uniqueness descendant codes,
we know that $t_0 = |{\mathcal{C}_0}|$ according to Proposition \ref{propo-multi-No-coll},
where $t_0$ is the number in Line $1$ of Algorithm \ref{alg-two-stage-STA}.
The algorithm below will perform iterations of the while loop.

In the first iteration,
since ${\bf x} = {\sf AT}(\mathcal{C}_0)$,
we know that ${\bf x}^{(i)}$ in Line $7$ of Algorithm \ref{alg-two-stage-STA}
is the generated word of the multiset $[{f({\bf b}_j(i))} \ | \ {\bf b}_j\in \mathcal{B}_0]$
for any $i \in \{1,2,\ldots,n\}$.
In addition, since $|\mathcal{C}_0|=t_0$, we have $|\mathcal{B}_0|=t_0$.
Thus $|[{f({\bf b}_j(i))} \ | \ {\bf b}_j\in \mathcal{B}_0]|=t_0$.
According to Proposition \ref{lem-find-multi-STA},
we know that Algorithm \ref{alg-multi-soft-TA} outputs $[{f({\bf b}_j(i))} \ | \ {\bf b}_j\in \mathcal{B}_0]$,
i.e., $[{\bf d}_j \ | \ j\in[U']] = [{f({\bf b}_j(i))} \ | \ {\bf b}_j\in \mathcal{B}_0]$,
where $[U']$ is the set in Line $8$ of Algorithm \ref{alg-two-stage-STA}.
So ${\bf R}(i)$ in Line $10$ of Algorithm \ref{alg-two-stage-STA} is in fact $\mathcal{B}_0(i)$.
Then ${\bf R}$ in Line $17$ of Algorithm \ref{alg-two-stage-STA} is ${\sf desc}(\mathcal{B}_0)$.
According to Lemma \ref{lem-find-inter},
Algorithm \ref{alg-q-find-inter} outputs an index set of a non-empty subset of $\mathcal{B}_0$
as $\mathcal{B}$ has $t$-uniqueness descendant code,
i.e.,  $U'' \neq \emptyset$ and $\{{\bf b}_j \ | \ j \in U''\} \subseteq \mathcal{B}_0$,
where $U''$ is the set in Line $17$ of Algorithm \ref{alg-two-stage-STA}.
According to the relation $\mathcal{C}_0=\mathcal{B}_0\circ \mathcal{D}$,
one can immediately derive that
\begin{eqnarray*}
U'' \neq \emptyset \  \mbox {and} \ \mathcal{C}_1 = \{{\bf c}_j \ | \ j \in U''\} \subseteq \mathcal{C}_0.
\end{eqnarray*}
In summary, during the first iteration,
one can obtain an index set of a non-empty subset $\mathcal{C}_1$ of $\mathcal{C}_0$
by using the generated word ${\bf x}={\sf AT}(\mathcal{C}_0)$.
\begin{itemize}
  \item If $|U|=t_0$, we can know that $|\mathcal{C}_1|=t_0=|\mathcal{C}_0|$.
Thus $\mathcal{C}_1=\mathcal{C}_0$ as $\mathcal{C}_1\subseteq \mathcal{C}_0$.
That is $\{{\bf c}_j \ | \ j \in U\}=\mathcal{C}_0$.
Therefore,   $\phi({\bf x})=\mathcal{C}_0$.
\item If $|U|< t_0$, then the algorithm proceeds to the second iteration.
\end{itemize}

In the second iteration, one can directly check that
the updated generated word ${\bf x}$ in the left-hand side of the equation  in Line $5$
is the generated word of $\mathcal{C}_0 \setminus \mathcal{C}_1$, i.e.,
${\bf x} = {\sf AT}(\mathcal{C}_0 \setminus \mathcal{C}_1)$.
Similar to the first iteration, we have
\begin{eqnarray*}
U'' \neq \emptyset \ \mbox{and} \ \mathcal{C}_2=\{{\bf c}_j \ | \ j\in U''\} \subseteq
\mathcal{C}_0 \setminus \mathcal{C}_1,
\end{eqnarray*}
where $U''$ is the set in Line 17 of Algorithm \ref{alg-two-stage-STA}.
\begin{itemize}
  \item If $|U|=t_0$, we can know that $|\mathcal{C}_2| +|\mathcal{C}_1| = t_0 = |\mathcal{C}_0|$,
  i.e., $|\mathcal{C}_2|=|\mathcal{C}_0|-|\mathcal{C}_1|$.
Thus $\mathcal{C}_2=\mathcal{C}_0 \setminus \mathcal{C}_1$
as $\mathcal{C}_2\subseteq\mathcal{C}_0  \setminus \mathcal{C}_1$.
Thus $\mathcal{C}_1 \cup \mathcal{C}_2 = \mathcal{C}_0$,
i.e., $\{{\bf c}_j \ | \ j \in U\}=\mathcal{C}_0$.
Therefore,  $\phi({\bf x})=\mathcal{C}_0$.
\item If $|U|< t_0$, then the algorithm proceeds to the next iteration.
\end{itemize}

Repeat this process.

In the last iteration,  we can obtain that
the updated generated word ${\bf x}$ in the left-hand side of the equation  in Line $5$
is the generated word of $\mathcal{C}_0 \setminus \cup_{i=1}^{s-1} \mathcal{C}_i$, i.e.,
 ${\bf x} = {\sf AT}(\mathcal{C}_0 \setminus \cup_{i=1}^{s-1} \mathcal{C}_i)$.
Thus
\begin{eqnarray*}
U'' \neq \emptyset \ \mbox{and} \ \mathcal{C}_s=\{{\bf c}_j \ | \ j\in U''\} \subseteq
 \mathcal{C}_0 \setminus \cup_{i=1}^{s-1}\mathcal{C}_i,
\end{eqnarray*}
where $U''$ is the set in Line 17 of Algorithm \ref{alg-two-stage-STA}.
Since this is the last iteration, $|U|=t_0$ must hold.
Then $|\mathcal{C}_s| +|\cup_{i=1}^{s-1}\mathcal{C}_i| = t_0 = |\mathcal{C}_0|$,
  i.e., $|\mathcal{C}_s|=|\mathcal{C}_0|-|\cup_{i=1}^{s-1}\mathcal{C}_i|$.
Thus $\mathcal{C}_s=\mathcal{C}_0 \setminus \cup_{i=1}^{s-1}\mathcal{C}_i$
as $\mathcal{C}_2\subseteq\mathcal{C}_0  \setminus \cup_{i=1}^{s-1}\mathcal{C}_i$.
Thus $\cup_{i=1}^{s}\mathcal{C}_i = \mathcal{C}_0$,
i.e., $\{{\bf c}_j \ | \ j \in U\}=\mathcal{C}_0$.
Therefore,  $\phi({\bf x})=\mathcal{C}_0$.
This completes the proof.

\begin{remark} According to Theorem \ref{th-two-sta-STA}, we need two codes with $t$-uniqueness descendant code.
That is both $\mathcal{B}$ and $\mathcal{D}$ have $t$-uniqueness descendant codes.
Recall Lemma \ref{lem-SSC-t-unique}, Corollary \ref{cor-FPC-t-unique} and Lemma \ref{lem-SMIPPC-unique},
each of $t$-FPCs, $\overline{t}$-SSCs and $t$-SMIPPCs has this property.
Hence one can select one or two types of such codes to
form concatenated codes $\mathcal{C}$.
\end{remark}

According to Definition \ref{def_AACC},
any concatenated  code $\mathcal{C}$ in Theorem \ref{th-two-sta-STA}
with Algorithm \ref{alg-two-stage-STA} is a $t$-AACC.
So we can again extend  Table \ref{tab-rela-all-codes} to Table \ref{tab-rela-all-codes-concate} by adding the concatenated  code $\mathcal{C}$.

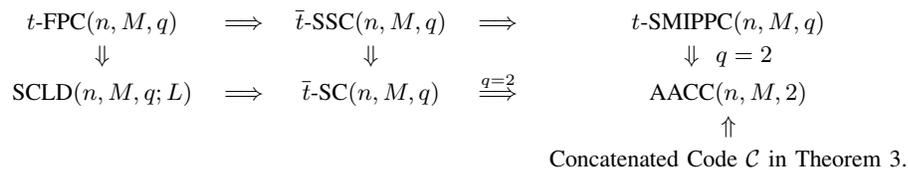
{\begin{table}[h]
\begin{center}
\caption{ Relationships among different types of AACCs}
\label{tab-rela-all-codes-concate}
   \begin{tabular}{ccccccccc}
     &  & $t$-FPC$(n,M,q)$ &  $ \Longrightarrow $ & $\overline{t}$-SSC$(n,M,q)$ & $\Longrightarrow$ & $t$-SMIPPC$(n,M,q)$ \\
     & & $ \Downarrow $  &  & $ \Downarrow $ &  & $\Downarrow$ \begin{footnotesize} $q=2$ \end{footnotesize}\\
     & & SCLD$(n,M,q; L)$  & $ \Longrightarrow $ & $\overline{t}$-SC$(n,M,q)$ &
     $ \overset{q=2}{\Longrightarrow}  $ & AACC$(n,M,2)$ \\
     &&&&&& $\Uparrow$\\
    &&&&&& Concatenated Code $\mathcal{C}$ in Theorem \ref{th-two-sta-STA}.
  \end{tabular}
\end{center}
\end{table}}

\begin{remark}
Up to now, we have showed  $t$-AACCs include $t$-FPC$(n,M,2)$,
$\overline{t}$-SSC$(n,M,2)$, $\overline{t}$-SCLD$(n,M,2)$, $\overline{t}$-SC$(n,M,2)$,
$t$-SMIPPC$(n,M,2)$ and the concatenated code $\mathcal{C}$ in Theorem \ref{th-two-sta-STA} as special cases.
The authors believe such known codes capture only a fraction of $t$-AACCs.
So it is interesting to discover new $t$-AACCs.
\end{remark}

We also extend  Table \ref{tab-comple-all-codes} to
Table \ref{tab-comple-AACCs} by adding  the concatenated code $\mathcal{C}$ in Theorem \ref{th-two-sta-STA}.
{\begin{table*}[h]
\center
\caption{Comparison of computational complexities of tracing algorithms of different types of AACCs}
\label{tab-comple-AACCs}
  \begin{tabular}{|c|c|c|c|c|c|c|c|c|}
\hline
 & FPC &   SSC  &   SC &  SCLD  & SMIPPC & Concatenated Code \\ \hline
\tabincell{c}{Complexity}
    &   $O(nM)$
    &   $O(nM)$
    &     $O(nM^t)$
    &   $O(\max\{nM,nL^t\})$
    &  $O(tnM)$
    &  $O(t^2nq + \min\{t,q\}tn_1M)$ \\  \hline

Reference & \cite{Bla,CM} &  \cite{JCM}  &   \cite{CJM,CM}
&    \cite{GVM}& \tabincell{c}{ Theorem \ref{th-soft-TA}, \\ Lemma \ref{lem-SMIPPC-unique} }&
\tabincell{c}{Theorem \ref{th-two-sta-STA},  Lemma \ref{lem-SSC-t-unique},
\\ Corollary \ref{cor-FPC-t-unique},  Lemma \ref{lem-SMIPPC-unique}}\\ \hline
\end{tabular}
\end{table*}

\begin{remark}
\label{rema-complex}
According to Table \ref{tab-comple-AACCs}, the following statements are always hold.
\begin{itemize}
  \item[1)] FPC, SSC and SCLD: When $t^2q <M$ and $\min\{t,q\}t < n_2$,
  noting that $n=n_1n_2$ in the concatenated code,
  we can derive that the concatenated code
  has better traceability than a $t$-FPC$(n, M, 2)$ (or a $\overline{t}$-SSC$(n, M, 2)$,
  or a $\overline{t}$-SCLD$(n, M, 2; L)$).
\item[2)] SC: It is obvious that the concatenated code  has better traceability than a $\overline{t}$-SC$(n, M, 2)$.
\item[3)] SMIPPC: When $tq <M$ and $\min\{t,q\} < n_2$,
the concatenated code  has better traceability than a $t$-SMIPPC$(n, M, 2)$.
\end{itemize}
\end{remark}

\section{Conclusion}\label{sec-conclu} %

In this paper, we proposed a research framework for multimedia fingerprinting codes that
designing algorithms first and then identifying codes compatible with these algorithms.
Specifically, we introduced the soft tracing algorithm and the two-stage soft tracing algorithm,
and showed that binary SMIPPCs and their concatenated codes satisfy the conditions of the above two algorithms, respectively.
Both theoretical and numerical comparisons shows that SMIPPCs achieve higher code rates.
It  would be  interesting to find out more promising algorithms and their corresponding codes.
}


\vskip 1cm


\begin{thebibliography}{1}


\bibitem{Alon2004}
N. Alon and U. Stav, ``New bounds on parent-identifying codes: The case of multiple parents," \emph{Combin., Probab. Comput.,} vol. 13, no. 6, pp. 795--807, 2004.





\bibitem{BCEKZ}
A. Barg, G. Cohen, S. Encheva, G. Kabatiansky, and G. Z\'{e}mor,
``A hypergraph approach to the identifying parent property: The case of multiple parents,"
{\it SIAM J. Discr. Math.},
vol. 14, no. 3, pp. 423-431, 2001.





\bibitem{BT}
M. Bazrafshan and T. van Trung,
``On optimal bounds for separating hash families,"
presented at the Germany Africa Workshop on Inform. Commun. Tech., Essen, Germany, 2008.





\bibitem{Bla}
S. R. Blackburn,
``Frameproof codes,"
{\it SIAM J. Discrete Math.},
vol. 16, no. 3, pp. 499-510, 2003.


\bibitem{Bl2}
S. R. Blackburn,
``Probabilistic existence results for separable codes,"
{\it IEEE Trans. Inform. Theory},
vol. 61, no. 11, pp. 5822--5827, 2015.










\bibitem{BS}
D. Boneh and J. Shaw,
``Collusion-secure fingerprinting for digital data,"
{\it IEEE Trans. Inform. Theory},
vol. 44, no. 5, pp. 1897-1905, 1998.

\bibitem{CW}
B. Chen and G. W. Wornell,
``Quantization index modulation: A class of provable good methods for digital watermarking and information embedding,"
{\it IEEE Trans. Inform. Theory},
vol. 47, no. 4, pp. 1423-1443, 2001.

\bibitem{Cheng}
M. Cheng,
Anti-Collusion Codes and Tracing Algorithms for Multimedia Fingerprinting,
Ph.D. dissertation, University of Ttukuba, 2012.

%
%


\bibitem{CFJLM}
M. Cheng, H. L. Fu, J. Jiang, Y. H. Lo,  and Y. Miao,
``Codes with the identifiable parent property for multimedia fingerprinting,"
{\it Des.  Codes  Cryptogr.},
vol. 83, no. 1, pp. 71-82, 2017.



\bibitem{CFJLM1}
M. Cheng, H. L. Fu, J. Jiang, Y. H. Lo,  and Y. Miao,
``New bounds on $\overline{2}$-separable codes of length 2,"
{\it Des.  Codes  Cryptogr.},
vol. 74, no. 1, pp. 31-40, 2015.




\bibitem{CJM}
M. Cheng, L. Ji, and Y. Miao,
``Separable codes,"
{\it IEEE Trans. Inform. Theory},
vol. 58, no. 3, pp. 1791-1803, 2012.

\bibitem{CM}
M. Cheng and Y. Miao,
``On anti-collusion codes and detection algorithms for multimedia fingerprinting,"
{\it IEEE Trans. Inform. Theory},
vol. 57, no. 7, pp. 4843-4851, 2011.

\bibitem{CKLS}
 I. J. Cox, J. Kilian, F. T. Leighton, and T. G. Shamoon,
 ``Secure spread spectrum watermarking for multimedia,"
 {\it IEEE Trans. Image Process.},
 vol. 6, no. 12, pp. 1673-1687, 1997.


\bibitem{DSSSU}
J. Dittmann, P. Schmitt, E. Saar, J. Schwenk, and J. Ueberberg,
``Combining digital watermarks and collusion secure fingerprints for digital images,"
 {\it SPIE J. Electron. Imag.},
vol. 9, no. 4, pp. 456-467, 2000.


\bibitem{EKK}
F. Ergun, J. Kilian, and R. Kumar,
``A note on the limits of collusion-resistant watermarks,"
in {\it Proc. Eurocrypt},
1999, pp. 140-149.

\bibitem{EFKL}
 E. Egorova, M. Fernandez, G. Kabatiansky, and M. H. Lee,
 ``Signature codes for A-channel and collusion-secure multimedia fingerprinting codes,"
 in {\it Proc. 2016 IEEE Int. Symp.  Inform. Theory},
2016, pp. 3043-3047.


\bibitem{For}
G. D. Forney,
Concatenated Codes, Cambridge, MA: MIT Press, 1966.


\bibitem{FGC}
 T. Furon, A. Guyader, and F. C\'{e}rou,
 ``Decoding fingerprinting using the markov chain monte carlo method,"
 in {\it IEEE Workshop on inform. Foren. Sec.},
 2012.



\bibitem{GG}
F. Gao and G. Ge,
``New bounds on separable codes for multimedia fingerprinting,"
{\it IEEE Trans. Inform. Theory},
vol. 60, no. 9, pp. 5257-5262, 2014.


\bibitem{GVM}
Y. Gu, I. Vorobyev, and Y. Miao
``Secure codes with list decoding,"
{\it IEEE Trans. Inform. Theory},
vol. 70, no. 4, pp. 2430-2442, 2024.

\bibitem{GST}
C. Guo, D. R. Stinson, and T. van Trung,
``On tight bounds for binary frameproof codes,"
{\it Des.  Codes  Cryptogr.},
vol. 77, pp. 301-319, 2015.




\bibitem{HLLT}
H. D. L. Hollmann, J. H. van Lint, J. P. Linnartz, and L. M. G. M. Tolhuizen,
``On codes with the identifiable parent property,"
{\it J. Combin. Theory Ser. A},
vol. 82, no. 1, pp. 121-133, 1998.


\bibitem{JCM}
J. Jiang, M. Cheng, and Y. Miao,
``Strongly separable codes,"
{\it Des.  Codes  Cryptogr.},
vol. 79, no. 2, pp. 303-318, 2016.


\bibitem{JGC}
J. Jiang, Y. Gu, and M. Cheng,
``Multimedia IPP codes with efficient tracing,"
{\it Des.  Codes  Cryptogr.},
vol. 88, no. 5, pp. 851-866, 2020.

\bibitem{JPWCH}
J. Jiang, F. Pei, C. Wen, M. Cheng, and H. D. L. Hollmann,
``Constructions of  $t$-strongly multimedia IPP codes with length $t+1$,"
{\it Designs, Codes and Cryptography},
vol. 92, no.10, pp. 2949-2970, 2024.







\bibitem{KLMSTZ}
J. Kilian, T. Leighton, L. Matheson, T. Shamoon, R. Tarjan, and F. Zane,
``Resistance of digital watermarks to collusive attacks,"
Dept. Comput. Sci., Princeton Univ., Princeton, NJ, Tech. Rep. TR-585-98, 1998.


\bibitem{LWLPF}
Q. Li, X. Wang, Y. Li, Y. Pan, and P. Fan,
``Construction of anti-collusion codes based on cover-free families,"
in {\it Proc. 6th Int. Conf. Inform. Tech.: New Generations},
Las Vegas, NV, 2009, pp. 362-365.

\bibitem{LT}
Z. Li and W. Trappe,
``Collusion-resistant fingerprints from WBE sequence sets,"
in {\it Proc. IEEE Int. Conf. Commun.},
Seoul, Korea, 2005, vol. 2, pp. 1336-1340.



\bibitem{LTWWZ}
K. J. R. Liu, W. Trappe, Z. J. Wang, M. Wu, and H. Zhao,
{\it Multimedia Fingerprinting Forensics for Traitor Tracing},
New York: Hindawi, 2005.








\bibitem{PZ}
C. I. Podilchuk and W. Zeng,
``Image-adaptive watermarking using visual models,"
{\it IEEE J. Select. Areas Commun.},
vol. 16, no. 4, pp. 525-539,  1998.









\bibitem{SWGM}
C. Shangguan, X. Wang, G. Ge, and Y. Miao
``New bounds for frameproof codes,"
{\it IEEE Trans. Inform. Theory},
vol. 13, no. 11, pp. 7247-7252,  2017.



\bibitem{SSW}
J. N. Staddon, D. R. Stinson, and R. Wei,
``Combinatorial properties of frameproof and traceablity codes,"
{\it IEEE Trans. Inform. Theory},
vol. 47, pp. 1042-1049,  2001.



\bibitem{Sto}
H. S. Stone,
Analysis of Attacks on Image Watermarks with Randomized Coefficients,
NEC Res. Inst., Princeton, NJ, 1996,  Tech. Rep. 96-045.

\bibitem{SEG}
J. K. Su, J. J. Eggers, and B. Girod,
``Capacity of digital watermarks subjected to an optimal collusion attack,"
in {\it Proc. Eur. Signal Process. Conf.},
2000.








\bibitem{TWL}
W. Trappe, M. Wu, and K. J. R. Liu,
``Collusion-resistant fingerprinting for multimedia,"
in {\it Proc. IEEE Int. Conf. Acoustics, Speech, Signal Process.},
Orlando, FL, 2002, pp. 3309-3312.

\bibitem{TWWL}
W. Trappe, M. Wu, Z. J. Wang, and K. J. R. Liu,
``Anti-collusion fingerprinting for multimedia,"
{\it IEEE Trans. Signal Process.},
vol. 51, no. 4, pp. 1069-1087, 2003.







\end{thebibliography}
\end{document}